\documentclass[12pt, a4paper]{article}
\usepackage[utf8]{inputenc}
\usepackage{graphicx}
\usepackage{amsmath, amsthm, amssymb, amsfonts}
\usepackage[top=2cm, bottom=2cm, left=2cm, right=2cm, marginparwidth = 2cm]{geometry}
\usepackage{hyperref}
\usepackage{sectsty}
\usepackage[justification=centerlast,labelsep=colon,labelfont=sc]{caption}
\usepackage[backend=bibtex, style=authoryear, citestyle=authoryear-comp, natbib=true, sorting=nyvt, doi=false,isbn=false,url=false]{biblatex}
\usepackage{xcolor}
\usepackage{geometry}
\usepackage{pdflscape}
\usepackage{setspace}

\renewbibmacro{in:}{}

\newtheorem{theorem}{Theorem}[section]
\newtheorem{lemma}{Lemma}
\newtheorem{corollary}{Corollary}[theorem]
\newtheorem{remark}{Remark}[theorem]

\newtheorem{assumption}{Assumption}[section]

\renewcommand\epsilon{\varepsilon}

\sectionfont{\centering \sc}
\subsectionfont{\centering \sc}
\subsubsectionfont{\centering \sc}
\allsectionsfont{\centering \sc}

\numberwithin{equation}{section}
\numberwithin{lemma}{section}

\title{Spectral Targeting Estimation of $\lambda-$GARCH models}
\date{\today}
\author{Simon Hetland\thanks{I am grateful to Anders Rahbek, Rasmus S\o ndergaard Pedersen for their support, guidance and valuable comments. This paper also benefited from feedback from Patrick Th\"oni, Andreas Hetland and Anne Lundgaard Hansen. Part of this research was conducted while I visited Imperial College Business School, for this I would like to thank Paolo Zaffaroni for his hospitality. %Code implementing the estimator is available from https://sites.google.com/site/simonlundhetland/ in MatLab (R2019a) and Python 3.0. The website also contains replication material for the empirical parts of the paper.  
Address correspondence to Simon Hetland, Department of Economics, University of Copenhagen, \O ster Farigmagsgade 5, Building 26, 1353 Copenhagen, Denmark. Email: slh@econ.ku.dk }} 
\begin{document}
\maketitle
%\newpage

\abstract{
This paper presents a novel estimator of orthogonal GARCH models, which combines (eigenvalue and -vector) targeting estimation with stepwise (univariate) estimation. We denote this the spectral targeting estimator. This two-step estimator is consistent under finite second order moments, while asymptotic normality holds under finite fourth order moments. The estimator is especially well suited for modelling larger portfolios: we compare the empirical performance of the spectral targeting estimator to that of the quasi maximum likelihood estimator for five portfolios of 25 assets. The spectral targeting estimator dominates in terms of computational complexity, being up to 57 times faster in estimation, while both estimators produce similar out-of-sample forecasts, indicating that the spectral targeting estimator is well suited for high-dimensional empirical applications.
}
%
%Keywords:
% Asymptotic theory, Multivariate GARCH, Variance targeting, Time series, Two-step estimation
%
%MSC classification:
%62M10 - Time series, GARCH
%62F12 - Asymptotic properties of parametric estimators
%62E20 - Asymptotic distribution theory in statistics
%62P05 - Applications of statistics to actuarial sciences and financial mathematics
%62P20 - Applications of statistics to economics PROB NOT THIS
%
%JEL classification
%C32 - Time series
%C58 - Financial Econometrics
\\  \\
\noindent 
\textbf{Keywords:} Asymptotic theory, Multivariate GARCH, Variance targeting, Two-step estimation.
%\\
%\textbf{MSC classifications:}  62M10, 62F12, 62P05.
\\
\textbf{JEL classifications:} C32, C58.

\section{Introduction}
\label{sec:introduction}
Multivariate conditionally heteroskedastic (MGARCH) models are a popular tool for risk management and dynamic portfolio allocation, where forecasts of conditional covariance matrices play an important role. 
As well known, MGARCH models suffer from the ``curse of dimensionality'', making them difficult and time consuming to estimate for larger portfolios using quasi maximum likelihood (QML) techniques. Many practitioners and academics alike have therefore preferred using alternative estimation methods: Two popular choices are the variance targeting (VT) estimator and the equation-by-equation (EbE) estimator, see e.g. \citet{bauwens2006}.

In the context of orthogonal GARCH models, such as the conditional eigenvalue GARCH ($\lambda-$GARCH) model of \citet{Hetland2019}, we can combine the idea behind the two methods in what we denote the spectral targeting estimator (STE): By estimating the unconditional eigenvalues and -vectors using a sample moment estimator, the remainder of the parameters of the GARCH model may be estimated univariately in a stepwise manner, in which we target the unconditional eigenvalues and -vectors. This estimation procedure dramatically reduces the computational complexity of the optimization problem and speeds up numerical estimation compared to the QML estimator. 

In this paper, we derive the large-sample properties of this two step estimator. Numerical illustrations show that the estimator is superior to the QML estimator in cross-sections larger than 10 financial assets, being up to 57 times faster in estimation, while the out-of-sample forecasts from the QML and ST estimator are similar in portfolios of 25 assets.

In general, asymptotic theory of QML estimation in MGARCH models is well-understood (see e.g. \citet{francq2019} (chapter 11) for a review of existing theory), whereas less attention has been paid to alternative estimation methods. Large sample properties of the two step VT estimator are considered in \citet{Pedersen2014} and \citet{francq2014} for the BEKK model \citep{engle1995} and (extended) CCC model (\citet{bollerslev1990} and \citet{jeantheau1998}) respectively, while \citet{Francq2016} consider the two step EbE estimator for various MGARCH specifications. Both the VT and EbE estimators are two-step estimators, which are quite common in econometrics, see e.g. \citet{Newey1994}. The EbEE and VTE both aim at making high(er) dimensional estimation feasible, and do so in two distinct ways: The EbEE estimates univariate volatility models in a first step, and subsequently a (conditional) correlation dynamic in a second step, whereas the VTE estimates the unconditional covariance matrix using a moment estimator, followed by a joint (profiled) estimation of the volatility and covariance dynamics. The ST estimator is related to both, as we recover sample eigenvalues and -vectors from the unconditional covariance matrix, and estimate univariate dynamics for ``rotated'' (orthogonalized) returns in a second step. The resulting estimator is well-behaved and easily implemented: Because the $\lambda-$GARCH model is specified using the spectral decomposition, the (profiled) log-likelihood, conditional on the initial estimator, can be rewritten as a sum of orthogonal univariate log-likelihood functions, making stepwise estimation feasible. This also means that the ST estimator is, in terms of asymptotic theory, equivalent to the VT estimator of the $\lambda-$GARCH. Furthermore, by recovering the (constant conditional) eigenvectors we avoid having to parameterize the eigenvectors under the restriction of orthonormality.

Consistency of the ST estimator follows under mild conditions (finite second order moments), while asymptotic normality requires finite fourth order moments, which may be violated empirically in financial data. Hence, the estimator is useful in the sense that it helps circumvent numerical issues often associated with (estimation of) MGARCH processes, and will produce consistent estimates under mild assumptions. However, it may not be suitable for inference on the parameters of the MGARCH process because of the moment requirement. Both of these moment conditions for consistency and asymptotic normality stem from the first step estimator of the unconditional eigenvalues and eigenvectors, which uses the sample covariance estimator. This is in contrast to the joint QML estimator of the $\lambda-$GARCH, which requires fractional moments for consistency and finite $2+\delta$ moments, for $\delta>0$, for asymptotic normality \citep{Hetland2019}.
\\ \\ 
The remainder of the paper proceeds as follows: Section \ref{sec:OGARCH} introduces the $\lambda-$GARCH model and spectral targeting. Section \ref{sec:estimation} presents the two-step estimator and Section \ref{sec:asymptotics} presents novel asymptotic results and discuss practical considerations for implementation. Section \ref{sec:empirical_illustration} investigates the empirical properties of the estimator compared to the QML estimator. Finally, Section \ref{sec:conclusion} concludes. All proofs are relegated to the appendices. 

\subsection{Notation}
\label{sec:notation}
Some notation used throughout the paper. $\mathbb{R}$ denotes the real numbers, $\mathbb{N}$ the natural numbers, and $\mathbb{Z}$ the positive natural numbers. The absolute value of $a\in\mathbb{R}$ is denoted $|a|$. For $p,n\in \mathbb{Z}$, $I_p$ denotes the $(p\times p)$ identity matrix and $0_{n\times p}$ denotes a $n\times p$ matrix of zeros.
The vector $\text{vec}(A)$ stacks the columns of the matrix A. We use the ``diag'' operator in two ways: If $W$ is a $p\times 1$ vector, $\text{diag}(W)$ returns a $p\times p$ diagonal matrix with $W$ on the diagonal, and if $A$ is a $p\times p$ matrix, $\text{diag}(A)$ returns the diagonal of $A$ as a $p\times 1$ vector. The trace of a square matrix is denoted $\text{tr}(A)$, and the determinant $\det(A)$. Furthermore, denote by $\rho(A)$ the spectral radius of any square matrix $A$, i.e. $\rho(A)=\max\{|\tilde\lambda_i|: \ \tilde\lambda_i \text{ is an eigenvalue of A}\}.$ We use $||\cdot||$ as a matrix norm. Let $\odot$ denote the Hadamard product, with $A^{\odot2}= A\odot A$, and $A\otimes B$ denotes the Kronecker product between A and B, and note that $A^{\otimes2}=A\otimes A$. Elements of matrices or vectors are denoted by lower case letters, e.g. $a_{ij}$ is the $(i,j)'th$ element of the matrix $A$. We use three kinds of convergence of random variables, $\overset{a.s.}{\rightarrow}$ denotes almost sure convergence, $\overset{p}{\rightarrow}$ denotes convergence in probability and $\overset{D}{\rightarrow}$ denotes convergence in distribution.

\section{The $\lambda-$GARCH model}
\label{sec:OGARCH}
As in \citet{Hetland2019}, we focus on the class of O-GARCH models originally introduced by \citet{alexander1997}. The presented model has more general dynamics than the O-GARCH, allowing for eigenvalue-spillovers, and we denote this version of the model the Eigenvalue GARCH, or $\lambda-$GARCH for short. 

Let $X_t$ be a $p\times1$ vector of asset returns,
\begin{align}
	X_t& =H_t^{1/2}Z_t, 
	\label{eq:Xt}
\end{align}
where $t=1,\hdots, T$ and $Z_t$ is an $iid(0,I_p)$ sequence of random variables. $H_t^{1/2}=V\Lambda_t^{1/2}$ is the (asymmetric) matrix square root of the conditional covariance matrix, $H_t$ (following the literature on MGARCH models, see e.g. \citet{van2002} and \citet{lanne2007}).%\footnote{We follow the literature on factor MGARCH models and use the asymmetric matrix square root, $H_t^{1/2}=V\Lambda_t^{1/2}$, see e.g. \citet{van2002} and \citet{lanne2007}, as opposed to the symmetric matrix square root, $H_t^{1/2}=V\Lambda_t^{1/2}V'$. In practice this has little implication, the formula for the exact asymptotic covariance matrix changes slightly, whereas the numerical approximations proposed in \eqref{eq:num_approx1}-\eqref{eq:num_approx2} are unchanged.} 
, which is decomposed using the spectral theorem,
\begin{align}
	H_t & = V\Lambda_tV'.
\end{align}
$V=\begin{pmatrix}
	V_1 & V_2 & \hdots & V_p
\end{pmatrix}$ is an orthonormal matrix of eigenvectors, $VV'=I_p$, and $\Lambda_t$ is a diagonal matrix with time-varying eigenvalues, $\lambda_t$, on the diagonal,
\begin{align}
	\Lambda_t&=\text{diag}(\lambda_t).
\end{align}
The $p\times1$ vector of dynamic eigenvalues are assumed to follow a GARCH dynamic,%\footnote{The dynamics for the time-varying eigenvalues are derived using the general framework of dynamic conditional score models of \citet{creal2013}, see also \citet{Hetland2019}.}
\begin{align}
	\lambda_t &= W+AY_{t-1}^{\odot2}+B\lambda_{t-1}, \label{eq:lambda} \ \ \ \ 
\end{align}
where $Y_t=V'X_t$ are ``rotated'' (or orthogonalized) returns: The orthonormal matrix $V$ rotates the returns $X_t$ to be linearly independent with conditional covariance $\Lambda_t$. To ensure that the covariance matrix is positive definite for all $t\in \mathbb{Z}$, we restrict $w_i>0$, $a_{ij}\geq0$, and $b_{ij}\geq0$ for $i,j=1,\hdots p$. Furthermore, to facilitate stepwise estimation, we restrict $B$ to be a diagonal matrix, letting the $i$'th lagged eigenvalue enter in equation $i$. 

By Lemma \ref{lemma:stationarity} and \ref{lemma:moments} in Appendix \ref{appendix:stat_erg_mom}, the stochastic process $\{X_t\}_{t\in\mathbb{Z}}$ can be initiated from the invariant distribution such that it is covariance stationary if and only if $\rho(A+B)<1$. If this is the case, the unconditional covariance matrix, $H = V(X_t) = E[X_tX_t']$, exists almost surely and is given by,
\begin{align}
	H & = V\text{diag}(\lambda) V', \\
	\lambda & = (I_p-A-B)^{-1}W,
	\label{eq:repar_w}
\end{align}
where $\lambda=E[\lambda_t]$ is the vector of unconditional eigenvalues. 

To obtain the covariance targeting (alternatively eigenvalue targeting) $\lambda-$GARCH, we re-parameterize the model by substituting \eqref{eq:repar_w} into \eqref{eq:lambda},
\begin{align}
	\lambda_t = (I_p-A-B)\lambda +A Y_{t-1}^{\odot2}+B\lambda_{t-1}.
\end{align}
This implies that the $i$th rotated return is driven by an augmented GARCH(1,1) with spill-overs from the other squared rotated returns,
\begin{align}
	y_{i,t}&=\lambda_{i,t}^{1/2}z_{i,t},	\label{eq:yi}\\
	\lambda_{i,t} & = w_i+\sum_{j=1}^pa_{ij}y_{j,t-1}^2+b_i\lambda_{i,t-1}, \label{eq:lambda_i}
\end{align}
where $w_i = (1-b_i)\lambda_i-\sum_{j=1}^pa_{ij}\lambda_j$ for $i=1,\hdots,p$, and $y_{i,t}=V_i'X_t$. This specification is motivated by generality: it seems restrictive to assume that the conditional variance of a component is not influenced by the past of other components, and allowing for spill-overs between assets may improve the model fit and out-of-sample performance.

\section{Spectral targeting estimation}
\label{sec:estimation}
While theory for classical joint QMLE of O-GARCH type models have been considered in \citet{Hetland2019}, we consider spectral targeting estimation (STE). The stepwise estimation procedure examined in this paper makes estimation and inference for the $\lambda-$GARCH feasible, even in large systems, as long as the time series dimension dominates the cross-sectional dimension (\citet{ledoit2004,ledoit2012}). %\footnote{It is well-known that the sample covariance matrix is not well-behaved when $p$ approaches $T$, see e.g. \citet{ledoit2004,ledoit2012} who employ shrinkage techniques to estimate $H$. Recently, \citet{Engle2019} consider the out-of-sample properties of large-dimensional multivariate variance-targeting GARCH models, in which they combine non-linear shrinkage of the unconditional covariance matrix with the DCC model.}

Define $\upsilon = \text{vec}(V)$, i.e. the vector of stacked eigenvectors, such that
\begin{align}
 %\gamma = \text{vec}(H) \ 
\gamma = [\lambda', \upsilon']' \text{ and } \ 	\kappa^{(i)} = [a_{i1},\hdots,a_{ip},b_i]',
	\label{eq:kappa_i}
\end{align}
where $\gamma$ contain the eigenvalues and -vectors of the unconditional covariance matrix, $H$. Hence, $\gamma$ denote ``static'' and $\kappa^{(i)}$ the ``dynamic'' parameters of equation $i$, such that $\theta^{(i)}=[\gamma', \kappa^{(i)\prime}]'$ is the vector of parameters associated with the $i$th rotated return, $i=1,\hdots,p$, of size $e=p^2+2p+1$.  Likewise, define the parameter space $\Theta^{(i)} := \mathcal{L} \times \mathcal{V} \times \mathcal{K}^{(i)} \subset \mathbb{R}^{p}_{++}\times\mathbb{R}^{p^2}\times \mathbb{R}^{p+1}$ which is restricted such that $\rho(A+B)<1$, $W$ and $\lambda$ are element-wise strictly positive and eigenvectors are orthonormal, $VV'=I_p$, such that $H$ is positive definite and symmetric. %Throughout, we may use the notation $\lambda_{i,t}(\gamma,\kappa^{(i)})$ and $y_{j,t}(\gamma)$ to denote the dependency of the $i$th eigenvalue on $\kappa^{(i)}$ and $\gamma$ and the dependency of the $j$th rotated return on $\gamma$ respectively. 
The vector of all the parameters in the model is 
\begin{align*}
	\theta=[\gamma',\kappa^{(1)\prime},\hdots, \kappa^{(p)\prime}]',
\end{align*} 
which has $p(p+1)/2+p^2+p$ elements. To emphasize the dependence on the parameters in $\theta^{(i)}$, we restate the model for the $i'$th rotated return as,
\begin{align*}
	y_{i,t}(\gamma) & = \lambda_{i,t}(\gamma, \kappa^{(i)})z_{i,t} \\
	\lambda_{i,t}(\gamma,\kappa^{(i)}) & =  w_i+\sum_{j=1}^pa_{ij}y_{j,t-1}^2(\gamma)+b_i\lambda_{i,t-1}(\gamma,\kappa^{(i)}),
\end{align*}
which also explicitly states that the conditional eigenvalues are a non-linear function of the eigenvectors in $\gamma$, and linear in the dynamic parameters in $\kappa^{(i)}$. Furthermore, $H_t(\theta)=V\Lambda_t(\theta)V'$, such that the (constant conditional) eigenvectors only depend on $\gamma$, whereas the diagonal matrix of conditional eigenvalues depend on the full vector of parameters, $\theta$.
%To ease the notation, we may suppress the dependency on $\theta$, $\gamma$ and $\theta^{(i)}$ in the following, noting that $H_t=H_t(\theta)$, $Y_t=Y_t(\gamma)$ and $\lambda_{i,t}=\lambda_{i,t}(\theta^{(i)})$ etc. 

The STE consists of two steps: In the first step, we estimate $\gamma$ using a sample estimator. In the second step, the dynamic parameters of the model are estimated by univariate QMLE for each equation in \eqref{eq:yi}-\eqref{eq:lambda_i} for $i=1,\hdots,p$. This procedure yields the STE for equation $i$, denoted $\theta^{(i)}$, and based on the joint vector of parameters, $\theta$, the sequence of filtrated conditional covariance matrices, $H_t(\theta)$, can be recovered for $t=1,\hdots,T$.

\subsection{The moment estimator}
The first step of the STE utilizes the (strong) law of large numbers for strictly stationary and ergodic processes, and we estimate $H$ by the sample covariance matrix,
\begin{align}
	\text{vec}(\hat H) = \text{vec}\left(\frac{1}{T}\sum_{t=1}^TX_tX_t'\right).
	\label{eq:gamma}
\end{align}
If $X_t$ is covariance stationary and ergodic, $\hat H$ is a strongly consistent estimator for $H$ by the ergodic theorem. From $\hat H$ it is possible to recover the estimated eigenvalues, $\hat\lambda$, and estimated eigenvectors, $\hat V$, by solving the two equations, 
\begin{align}
& \left|H-\lambda I_p\right|=0,	 \\
& HV_i=\lambda_i V_i, \  \ \ \ \ \ \ \ \ \ \ i=1,\hdots,p,
\end{align} 
and under Assumption \ref{assumption1} below $\hat \lambda$ and $\hat \upsilon$ are strongly consistent estimators of $\lambda$ and $\upsilon$ respectively by the continuous mapping theorem. 

In applications, these two equations are solved using iterative procedures and for $H$ symmetric and positive definite, all eigenvalues are almost surely strictly positive. Notice however, that the eigenvalue decomposition is not unique: the spectrum of $H$ is unique only up to the ordering, and while the eigenspace of $H$ is unique the eigenvectors are not. Furthermore, eigenvalues may not be unique. We discuss this further in remark \ref{remark:identification_mom}.

\begin{remark}[Alternative first step estimator] \label{remark:alt_1_step}
Instead of estimating the eigenvalues and -vectors implicitly using the moment estimator of $H$, we can estimate them directly using an approach similar to that proposed by \citet{fan2008} and \citet{boswijk2011}, wherein $V$ is specified using rotation matrices,
\begin{align*}
	V(\phi)=\prod_{1\leq i< j \leq p}U_{ij}(\phi_{ij}),
\end{align*}
with $U_{ij}(\phi_{ij})$ a $p$-dimensional identity matrix apart from four elements: $(i,i)$ and $(j,j)$ are $\cos(\phi_{ij})$, $(i,j)$ and $(j,i)$ are $\sin(\phi_{ij})$ and $-\sin(\phi_{ij})$ respectively. $\phi$ is a $p(p-1)/2$ vector containing the rotation parameters, $\phi_{ij}$. This parameterization ensures that $V(\phi)V'(\phi)=I_p$.
The eigenvectors and eigenvalues can then be estimated by numerically solving the minimization problem,
\begin{align*}
	\underset{[\phi',\lambda']'\in\mathcal{C}}{\arg\min} \ \ C_T(\phi, \lambda)
\end{align*}
where $\mathcal{C}$ is an appropriate parameter space and $C_T(\phi,\lambda)$ is a cost function, e.g. the Gaussian log-likelihood,
\begin{align*}
	C_T(\phi,\lambda) = \frac{1}{T} \sum_{t=1}^T \left (\log\det(\Lambda)+X_t'V(\phi)\Lambda^{-1}V'(\phi)X_t\right).
\end{align*}
The asymptotic theory for this estimator can be derived with relative ease, see e.g. \citet{Hetland2019} who parameterize the joint QMLE of the $\lambda-$GARCH in a similar fashion. One should however, keep in mind that the rotation parameters in $\phi$ are not uniquely identified unless we impose restrictions on the parameter space. A sufficient condition is $\phi_{ij}\in(0,\pi/2)$.
\end{remark}
The alternative first step estimator outlined in remark \ref{remark:alt_1_step} requires numerical optimization of a cost function, and may therefore run into numerical problems as $p$ increase, such as failure of a Newton-type optimization procedure to converge, or the possibility of ending up in a local maximum -- problems similar to those of the joint QML estimator. We therefore choose to work with the sample moment estimator as it has a closed form solution and is the preferred first step estimator in the variance targeting literature.

\subsection{The profiled maximum likelihood estimator}
In the second step of the STE, we consider the profiled quasi log-likelihood function based on the multivariate Gaussian distribution. The joint Gaussian log-likelihood of the model, conditional on a fixed $X_0$ and $H_0$, is,
\begin{align}
	%L_T(\theta) & = \frac{1}{T}\sum_{t=1}^Tl_{t}(\theta), \\
    %l_t(\theta) & = \log\det(H_t)+X_t'H_t^{-1}X_t = \sum_{i=1}^p\log(\lambda_{i,t})+\frac{y_{i,t}^2}{\lambda_{i,t}}.
    L_T(\theta) & = \frac{1}{T}\sum_{t=1}^T \log\det(H_t(\theta))+X_t'H_t^{-1}(\theta)X_t \notag\\
    			& = \frac{1}{T}\sum_{t=1}^T\sum_{i=1}^p\left(\log(\lambda_{i,t}(\gamma,\kappa^{(i)}))+\frac{y_{i,t}^2(\gamma)}{\lambda_{i,t}(\gamma,\kappa^{(i)})}\right) \notag\\
    			& = \sum_{i=1}^p L_T^{(i)}(\gamma,\kappa^{(i)}), \label{eq:joint_loglik}
\end{align}
using $H_t^{(-1)}(\theta)=V\Lambda_t^{-1}(\theta)V'$, $\log\det(H_t(\theta))=\sum_{i=1}^p\log(\lambda_{i,t}(\gamma,\kappa^{(i)}))$. and $Y_t(\gamma)=V'X_t$.
That is, because the rotated returns are orthogonal, the log-likelihood function can be decomposed as the sum of $p$ univariate log-likelihood functions, each of which depend on $\theta^{(i)}=[\gamma', \kappa^{(i)'}]'$,
\begin{align}
	L_T^{(i)}(\gamma,\kappa^{(i)}) & = \frac{1}{T}\sum_{t=1}^T l_t^{(i)}(\gamma,\kappa^{(i)}), \\
	l_t^{(i)}(\gamma,\kappa^{(i)}) & = \log(\lambda_{i,t}(\gamma,\kappa^{(i)}))+\frac{y_{i,t}^2(\gamma)}{\lambda_{i,t}(\gamma,\kappa^{(i)})},	
\end{align}
where $y_{i,t}(\gamma)= V_i'X_t$ and $\lambda_{i,t}(\gamma,\kappa^{(i)})$ is given in \eqref{eq:lambda_i}. Conditional on $\gamma$, each of the $i$ univariate log-likelihood functions are orthogonal and do not depend on $\kappa^{(j)}$ for $j\neq i$. The parameters of the model can therefore be estimated sequentially, and we define the STE of $\kappa^{(i)}$ as,
\begin{align}
	\hat\kappa^{(i)}=\underset{\kappa^{(i)}\in \mathcal{K}^{(i)}}{\arg\min \ }L_{T}^{(i)}(\hat\gamma, \kappa^{(i)}),
	\label{eq:QMLE}
\end{align}
and the two-step procedure yields the STE of $\theta$,
\begin{align*}
	{\hat\theta} = [\hat\gamma',\hat\kappa^{(1)\prime},\hdots, \hat\kappa^{(p)\prime}]'.
\end{align*}

Similar to (quasi) maximum likelihood estimation of multivariate GARCH models, we use the Gaussian log-likelihood function, but we do not assume that the vector of innovations $Z_t$ are Gaussian, only that they are centered with unit variance: Even if the innovations are drawn from a different distribution, the results in Theorem \ref{theorem:consistency} and \ref{theorem:asympnorm} below still hold, as long as the assumptions are satisfied. 

Compared to joint QMLE, which estimates all $\frac{3}{2}(p^2+p)$ parameters jointly, the STE procedure vastly reduces the number of parameters estimated in each step: In the first step $p(p+1)/2$ parameters are estimated by method of moments and in the second step $p+1$ parameters are estimated for each rotated return, making the estimation procedure suitable in high-dimensional systems and less vulnerable to numerical problems.

\section{Large-sample properties of sequential variance targeting estimation}
\label{sec:asymptotics}
In this section we establish consistency and asymptotic normality of the STE and discuss practical considerations for implementation. A novelty of the asymptotic theory presented here is that we parameterize the moment estimator in terms of the unconditional eigenvalues and vectors, rather than the vectorized covariance matrix. In doing so, we apply the mean-value theorem on the eigenvectors, which otherwise do not have a closed form solution as a function of the unconditional covariance matrix. This, in conjunction with the continuous mapping theorem, allows us to study the asymptotic behavior of both the first step estimator, $\gamma$, and the joint parameter vector of the $i$'th rotated return, $\theta^{(i)}$. 

The two-step estimator is consistent under finite second order moments, and it has a limiting Gaussian distribution under the assumption of finite fourth order moments. Both of these moment conditions stem from the first step moment estimator, and are more strict that the moment conditions for the joint QML estimator (for which we need $E||X_t||^{2+\delta}<\infty$, $\delta>0$, see Theorem 3.3 in \citet{Hetland2019}). These results are novel and extend the existing literature on targeting and stepwise estimation, see e.g. \citet{francq2014}, \citet{Pedersen2014} and \citet{Francq2016}. All proofs are relegated to Appendix \ref{appendix:proofs}. 

Before discussing the asymptotic properties in detail, we make the following assumptions. First, we assume that the process is covariance stationary and ergodic.
\begin{assumption}
	The process $\{X_t\}_{t\in \mathbb{Z}}$ is strictly stationary, ergodic and has finite second order moments. 
	\label{assumption1} 
\end{assumption}
Furthermore, we need the following assumption on the algebraic multiplicity of the eigenvalues.
\begin{assumption}
    The characteristic polynomial of the unconditional covariance matrix, $H$, has an algebraic multiplicity of 1.
	\label{assumption1.5new} 
\end{assumption}
We also assume that the dynamic parameters of the model are identified and that the true parameter vector is a subset of the parameter space.
\begin{assumption}
	The true parameter vector $\theta_0^{(i)}\in\Theta^{(i)}$, with $\Theta^{(i)}$ compact.  \label{assumption2}
\end{assumption}
\begin{assumption}
	For $\kappa^{(i)}\in \mathcal{K}^{(i)}$, if $\kappa^{(i)}\neq\kappa_0^{(i)}$, then %$H_t(\gamma_0,\kappa)\neq H_t(\gamma_0,\kappa_0)$. \label{assumption3} }}
$\lambda_{i,t}(\gamma_0,\kappa^{(i)})\neq \lambda_{i,t}(\gamma_0,\kappa_0^{(i)})$. \label{assumption3} 
\end{assumption}

These assumptions lead us to the following theorem on strong consistency of the ST estimator.
\begin{theorem}
	Under Assumptions \ref{assumption1}-\ref{assumption3}, as $T\rightarrow\infty$, the ST estimator is consistent,
\begin{align*}
	\hat\theta^{(i)}\overset{a.s}{\rightarrow}\theta_0^{(i)}. 
\end{align*}\label{theorem:consistency}
\end{theorem}\vspace{-0.6cm}

Assumption \ref{assumption1} is in line with the literature for variance-targeting estimation, both in the univariate and multivariate case, see e.g. \citet{Pedersen2014} or \citet{francq2011,francq2014}, and is needed to ensure that the moment estimator converge to a well-defined unconditional covariance matrix for $T\rightarrow \infty$. 

%\begin{remark}[Eigenvalue multiplicity]
%\label{remark:simple_eig}
Assumption \ref{assumption1.5new} is novel in the (variance) targeting literature and is needed for the first step estimator: We assume that all the unconditional eigenvalues are simple, i.e. that the characteristic polynomial of the unconditional covariance matrix has an algebraic multiplicity of one. 
This is needed for two reasons: First, in the case of repeated eigenvalues, the associated eigenvectors are not uniquely determined and the parameters in the first step estimator are not uniquely identified, and hence the first step estimator is not consistent. Second, it is a requirement for $\lambda$ and $\upsilon$ to be continuously differentiable (Theorem 1, \citet{magnus1985}), which is needed to apply the mean-value theorem when considering the asymptotic distribution of the estimator. 
%\end{remark}
%As noted in remark \ref{remark:simple_eig}, the asymptotic results break down in the case of repeated eigenvalues as the mean-value theorem cannot be applied on the vector of eigenvectors, $\upsilon$, and the estimator of the eigenvalues and -vectors is not consistent. 
%However, we conjecture that the estimate of $\{H_t\}_{t\in \mathbb{Z}}$ is still consistent in the case of repeated eigenvalues.

Assumptions \ref{assumption2}-\ref{assumption3} are standard for multivariate GARCH models, see e.g. \citet{COMTE200361} or \citet{Hafner2009a}. 
Moreover, the normalization imposed on the first step estimator ensures that the eigenvalues and -vectors of the first step estimator are uniquely identified. Primitive conditions for the identification of the second step estimator can be found in e.g. \citet{francq2019} (chapter 10) and are also treated in \citet{Hetland2019}.

\begin{remark}[Identification of the first step estimator]\label{remark:identification_mom}
The eigenvalue decomposition used in the first step estimator is not uniquely defined: the ordering of the eigenvalues is not fixed and the sign of the eigenvectors is unidentified. We can, however, without a loss of generality, sort the eigenvalues in non-decreasing order and normalize the eigenvectors such that the first non-zero element of each eigenvector is positive. These two normalizations, along with Assumption \ref{assumption1.5new}, ensure that the eigenvalue decomposition is unique. Note, however, that any equivalent normalizations also suffice.
\end{remark}

Next, we show that the estimator is asymptotically normal. To do so, we need two additional assumptions on existence of moments and the true parameter vector.
\begin{assumption}
	The process $\{X_t\}_{t\in \mathbb{Z}}$ has finite fourth order moments, $E||X_t||^4<\infty$.  \label{assumption4}	
\end{assumption}
\begin{assumption}
	$\theta_0^{(i)}$ is in the interior of $\Theta^{(i)}$.  \label{assumption5} 
\end{assumption}

This leads us to the next theorem on asymptotic normality of the estimator for the $i$'th rotated return,
\begin{theorem} % \color{red}
	Under Assumptions \ref{assumption1}-\ref{assumption5}, for $T\rightarrow\infty$
\begin{align*}
\sqrt{T}\left(\hat\theta^{(i)}-\theta_0^{(i)}\right)\overset{D}{\rightarrow}N\left(0,\Sigma_0^{(i)}\right),
\end{align*}
where $\Sigma^{(i)}$ is the asymptotic covariance matrix, given by,
\begin{align}
	\underset{e\times e}{\Sigma_0^{(i)}} & = 	
	\begin{pmatrix}
		I_{p(p+1)} & 0_{p(p+1)\times p+1} \\
		-(J_0^{(i)})^{-1}K_0^{(i)} & -(J_0^{(i)})^{-1}
	\end{pmatrix}
	\Omega_0^{(i)}
	\begin{pmatrix}
		I_{p(p+1)} & -(J_0^{(i)})^{-1}(K_0^{(i)})' \\
		0_{p+1\times p(p+1)} & -(J_0^{(i)})^{-1}
	\end{pmatrix}.
	\label{eq:asymp_var_maintext}
\end{align}
where ${J_0^{(i)}}$ and ${K_0^{(i)}}$  are defined in \eqref{eq:appendix_asympvar} and ${\Omega_0^{(i)}}$ is given in \eqref{eq:asymp_var}.
\label{theorem:asympnorm}
\end{theorem}

Assumption \ref{assumption4} is required to ensure that the first step estimator, $\sqrt{T}(\hat\gamma-\gamma_0)$, converges to a Gaussian distribution with a finite variance. This assumption is common in the variance targeting literature and is also needed when reparameterizing the moment estimator in terms of the spectral decomposition. In fact, the moment requirement is not needed in the probability analysis of the profiled log-likelihood function, but it simplifies the exposition. Lemma \ref{lemma:moments} in Appendix \ref{appendix:stat_erg_mom} can be used to check the moment condition in Assumption \ref{assumption4}. Based on the simulations included in Appendix \ref{appendix:simulation}, the moment conditions for consistency and asymptotic normality are sufficient and necessary. Assumption \ref{assumption5} is standard in the literature, and is a technical requirement to ensure that the mean-value theorem can be applied on the optimality condition for the profiled log-likelihood functions.

In the derivation of the asymptotic distribution of the first step (and consequently the second step) estimator, we restate the moment estimator as the average of the conditional eigenvalues. However, as $\text{vec}(\hat H)=V^{\otimes2}\text{vec}(\frac{1}{T}\sum_{t=1}^T\Lambda_{0,t}^{1/2}Z_tZ_t'\Lambda_{0,t}^{1/2})$ is in terms of $\Lambda_t$, and not the vectorized eigenvalues, $\lambda_t$, we restate the dynamics of the conditional eigenvalues in \eqref{eq:lambda} as a (restricted) BEKK$(p^2,1,1,1)$ model for $\Lambda_t$. This parametrization is present in $\Omega_{0}^{(i)}$ in \eqref{eq:asymp_var}. In doing so, $\hat\gamma-\gamma_0$ is a martingale difference, allowing us to use a central limit theorem on $\sqrt{T}\begin{pmatrix}
	\hat\gamma-\gamma_0 & \frac{\partial L_T^{(i)}(\theta_0^{(i)})}{\partial \kappa^{(i)}} 
\end{pmatrix}'$
jointly to show normality and find the expression for $\Omega^{(i)}_0$. The proof of joint normality of $\sqrt{T}(\hat\theta^{(i)}-\theta^{(i)}_0)$ applies the mean-value theorem on the optimality condition of the second step estimator, stacked with the moment estimator from step one, and lemmata \ref{lemma:norm1}-\ref{lemma:norm5} in Appendix \ref{appendix:lemmas} verify that the mean-value theorem can be applied. 

\begin{remark}[Fixed initial values]
	Assumption \ref{assumption1} assumes that the process is strictly stationary, implying that the process $\{X_t\}_{t\in\mathbb{Z}}$ is initiated in the invariant distribution or in the infinite past.  In practice the observed process is initiated in some fixed values, $X_0$ and $H_0$, which by definition makes the process non-stationary. However, the appendix verifies that the choice of initial values are asymptotically irrelevant for both consistency and asymptotic normality of the estimator, see lemmata \ref{lemma:con4} and \ref{lemma:norm5}.
\end{remark}
The model as presented in \eqref{eq:Xt}-\eqref{eq:lambda} restricts the $B$ matrix to facilitate stepwise estimation. However, in many applications, practitioners prefer the ``diagonal'' specification, in which both $A$ and $B$ are diagonal, this restriction is also feasible in terms of estimation, as discussed in the following remark. 

\begin{remark}[Diagonal model]
 	The estimator is still consistent and asymptotically normal in the case of a diagonal $A$ matrix, such that the vector of dynamic parameters is $\kappa^{(i)}=[a_i,b_i]'$. Hence, by Theorem \ref{theorem:asympnorm}, Assumption \ref{assumption4} and Lemma \ref{lemma:moments} implies that for $A$ and $B$ diagonal matrices, a sufficient and necessary condition for finite fourth order moments is, $\max(b_i^2+E[z_{i,t}^4]a_{i}+2a_ib_i)<1$ for $i=1,\hdots,p$.
\end{remark}

In applications, the asymptotic variance matrix for the stepwise estimator may be approximated using plug-in sample estimators. That is,
\begin{align}
	\underset{p+1\times p+1}{\hat J_T^{(i)}} & = \frac{1}{T}\sum_{t=1}^T\frac{\partial^2 l_t^{(i)}(\theta^{(i)})}{\partial\kappa^{(i)}\partial\kappa^{(i)'}}\Bigg\vert_{\theta^{(i)}=\hat\theta^{(i)}}, \ \ \ \ \ \ 
	\underset{p+1\times p^2+p}{\hat K_T^{(i)}}  = \frac{1}{T}\sum_{t=1}^T\frac{\partial^2 l_t^{(i)}(\theta^{(i)})}{\partial\kappa^{(i)}\partial\gamma'}\Bigg\vert_{\theta^{(i)}=\hat\theta^{(i)}}, \label{eq:num_approx1}\\
	\underset{e\times e}{\hat\Omega_T^{(i)}} & = \frac{1}{T}\sum_{t=1}^T \hat\omega_t\hat\omega_t'
	\ \ \ \text{ where } \ \ \ 
	\underset{e\times1}{\hat\omega_t} := 
	\begin{pmatrix}
		&D \hat V^{\otimes2}\left(\text{vec}(X_tX_t)-\text{vec}(\hat H)\right) \\
		&\hat V_1'\otimes (\hat\lambda_1 I_p-\hat H)^+\left(\text{vec}(X_tX_t)-\text{vec}(\hat H)\right) \\
		&\vdots \\
		&\hat V_p'\otimes (\hat\lambda_p I_p-\hat H)^+\left(\text{vec}(X_tX_t)-\text{vec}(\hat H)\right) \\
		&\frac{\partial l_t^{(i)}(\theta^{(i)})}{\partial \kappa^{(i)}}\Big\vert_{\theta^{(i)}=\hat\theta^{(i)}}
	\end{pmatrix}
	\label{eq:num_approx2}
\end{align}
where $(\hat\lambda_i I_p-\hat H)^+$ denotes the Penrose-Moore pseudo-inverse of $\hat\lambda_i I_p-\hat H$ for $i=1,\hdots,p$ and $D$ is defined in Lemma \ref{lemma:norm0}. The expressions in \eqref{eq:num_approx1}-\eqref{eq:num_approx2} converge almost surely to their population counterparts due to the ergodic theorem. Note that we may substitute the estimators by their true value, as we have established strong consistency of $\theta^{(i)}$. This makes estimation of the asymptotic covariance matrix only slightly more cumbersome than that of the well-known ``sandwich'' covariance matrix estimator known from joint QMLE.%\footnote{Code implementing the SVTE procedure for the $\lambda-$GARCH along with the estimator of the sample variance-covariance matrix for all parameters is available from https://sites.google.com/site/simonlundhetland/ implemented in MatLab (R2019a) and Python 3.0.} 

Once we know the asymptotic distribution of $\hat\theta^{(i)}$, it is possible to derive the asymptotic distribution of the intercept term, $w_i$, in the original model in \eqref{eq:lambda_i} using the delta method.

\begin{corollary}[Limiting distribution of the intercept]
Under Assumption \ref{assumption1}-\ref{assumption5}, for $T\rightarrow \infty$,
%\begin{align*}
%	& \sqrt{T}\left(\hat\lambda-\lambda_0\right)\overset{D}{\rightarrow}N(0,\Phi_0\Sigma_0^{(i)}\Phi_0'), \\
%	& \sqrt{T}\left(\hat V_j-V_{0,j}\right)\overset{D}{\rightarrow}N(0,\Xi_0\Sigma_0^{(i)}\Xi_0'), 
%\end{align*}
%for $j=1,\hdots,p$ and
\begin{align*}
	& \sqrt{T}\left(\hat w_i-w_{0,i}\right)\overset{D}{\rightarrow}N(0,\Phi_0\Sigma_0^{(i)}\Phi_0'),
\end{align*}
for $i=1,\hdots,p$ with $\Sigma_0^{(i)}$ given in \eqref{eq:asymp_var_maintext} and,
\begin{align}
	\underset{1\times e}{\Phi_0} = 
	\begin{pmatrix}
		1-b_{0,i}\mathbb{I}\{\lambda_{0,i}\}-\sum_{j=1}^p a_{0,ij}, &
		0_{p^2\times 1}, &
		-\lambda_0, &
		-\lambda_{0,i}
	\end{pmatrix},
\end{align}
where $\mathbb{I}\{\lambda_i\}$ is a $p\times 1$ vector of zeros, apart from a $1$ in the $i$'th row.
\label{corollary:dist_w_l}
\end{corollary}

Note that we present the asymptotic theory in terms of $\theta^{(i)}$ rather than, $\theta$, following \citet{Francq2016} and their notation for an equation-by-equation estimator for various MGARCH models. The theorems listed above could easily be restated in terms of $\theta$, as the asymptotic results hold simultaneously due to the orthogonality of the conditional univariate log-likelihood functions, but we refrain from doing so for two reasons: First, the present formulation is coherent with the step-wise approach of the estimator. Second, the present formulation makes it straightforward to parallelize estimation, computing the asymptotic variance matrix in each iteration, which also speeding up the estimation procedure. 

As already emphasized, the STE reduces the risk of numerical issues in estimation compared to the QML estimator, and in the context of the $\lambda-$GARCH model, the ST estimator is closely related to the variance targeting estimator: Because the profiled log-likelihood consists of $p$ orthogonal terms, the STE and VTE of the $\lambda-$GARCH are theoretically equivalent, and in practice is expected to produce similar estimates and standard errors. Note however, that we expect the STE to have a smaller computational burden, as it minimizes the log-likelihood function over a smaller parameter space.
\begin{remark}[Variance targeting estimator of the $\lambda-$GARCH]
	An alternative to the stepwise estimation of the $\lambda-$GARCH is the variance targeting (VT) estimator, in which all $\kappa=[\kappa^{(1)\prime},\hdots,\kappa^{(p)\prime}]'$ are estimated jointly. This estimator still relies on the first step estimator of $\hat\gamma$ in \eqref{eq:gamma}, and we denote the full VT estimator $\tilde\theta=[\hat\gamma',\tilde\kappa']'$. Because of the orthogonal structure of the log-likelihood function, consistency and asymptotic normality of the VT estimator of the $\lambda-$GARCH can be derived using similar techniques as in Appendix \ref{appendix:proofs} and \ref{appendix:lemmas}.
\end{remark}

\section{Empirical illustrations}
\label{sec:empirical_illustration}
In the following we compare the empirical performance of the ST estimator to that of the joint QML estimator. First, we consider the relative efficiency of the two estimators in a simulation setting for many different portfolio sizes. This exercise lets us compare (empirical) efficiency of the STE against the QMLE. Second, we consider the out-of-sample performance of the the two estimation methods in a recursive value-at-risk application for portfolios of $p=25$ assets. The empirical fit is assessed using the likelihood ratio tests of \citet{christoffersen1998}. In both these exercises, we also consider the computational complexity (i.e. time spent on estimating the model) of the two methods. Finally, we briefly summarize the results.

\subsection{Relative efficiency: STE vs. QMLE}
\label{sec:relative_efficiency}
We now compare the relative efficiency and the time complexity of the STE against the joint QMLE. This is done for the diagonal model of dimension $p$, where we simulate a data-generating process $N=399$ times with $A_0=0.05I_p$, $B_0=0.85I_p$, such that the process has finite fourth order moments. The unconditional eigenvalues are specified as $\lambda_{0,i}=(p+1-i)/10$ for $i=1,\hdots,p$ and the eigenvectors are constructed using rotation matrices with all rotation parameters $\phi_{0,i}=0.5$ for $i=1,\hdots,p(p-1)/2$ (see Remark \ref{remark:alt_1_step}). The innovations, $Z_t$, are drawn $iid$ from a standard normal distribution, and each path of the simulated process has $T=2000$ observations. The model has $p(p-1)/2+3p$ parameters, and the STE procedure estimates $p(p+1)/2$ parameters in the first step, and the remaining $2p$ parameters sequentially for each rotated return. The QMLE on the other hand estimates all $p(p-1)/2+3p$ parameters simultaneously. %In the numerical estimation, both estimators are initiated in the true parameter values minus 0.025.

In comparing the two estimators, we employ the same methodology as \citet{Francq2016} who use the quadratic form $T(\hat\vartheta-\vartheta_0)'\mathcal{I}(\hat\vartheta-\vartheta_0)$ as a measure of accuracy of an estimator $\hat\vartheta$, where $\mathcal{I}$ is the (numerically) approximated information matrix and the parameter vector is constructed identically for both estimators with $\vartheta = [\text{vec}(H)',\text{diag}(A)',\text{diag}(B)']'$. Because $\mathcal{I}$ is computationally demanding to compute in higher dimensions, we instead use the simulated information matrix, which is obtained as $\mathcal{I}=var(\hat\vartheta^n_{QMLE}-\vartheta_0)\approx \frac{1}{N}\sum_{n=1}^N(\hat\vartheta^n_{QMLE}-\vartheta_0)(\hat\vartheta^n_{QMLE}-\vartheta_0)'$ for $n=1,\hdots,N$, where $\hat\vartheta^n_{QMLE}$ is the QMLE parameter vector for the $n$'th simulated path. The relative efficiency is then computed as,
\begin{align*}
	RE = \frac{(\bar\vartheta_{STE}-\vartheta_0)'\mathcal{I}(\bar\vartheta_{STE}-\vartheta_0)}{(\bar\vartheta_{QMLE}-\vartheta_0)'\mathcal{I}(\bar\vartheta_{QMLE}-\vartheta_0)},
\end{align*}
where $\bar\vartheta_{STE}=\frac{1}{N}\sum_{n=1}^N\hat\vartheta^n_{STE}$ with $\bar\vartheta_{QMLE}$ defined analogously. By this definition, if $RE<1$, the ST estimator is relatively more efficient than the QML estimator.

% Table generated by Excel2LaTeX from sheet 'Latex table'
\begin{table}[!htb]
  \centering
  \caption{Relative efficiency: QMLE vs. STE}
    \begin{tabular}{ccccc}
    \hline
    \hline
    Dimension, $p$ & $\#$ parameters & Time (s), QMLE & Time (s), STE & RE \\
    \hline
       	2  	&    7  	&   12.82  		&   0.69  	 &  2.59  \\
   		4  	&   18  	&	52.93  		&   1.36  	 &  1.42  \\
	   	6  	&   33  	&   161.70  	&   1.92  	 &  2.05  \\
        8  	&   52  	&   358.16  	&   2.74  	 &  2.06  \\
        10  &   75  	&   505.75  	&   3.76  	 &  1.13  \\
        12  &  102  	&   546.40  	&   4.86  	 &  0.11  \\
        15  &  150  	&   617.31  	&   6.41  	 &  0.01  \\
        20  &  250  	&   765.87  	&   10.03  	 &  0.00  \\
        50 	&  1,375  	&	2,987.88  	&   33.00  	 &  0.03  \\
        100	&  5,250  	&   15,775.08  	&   141.27   &  0.04  \\
        200 &  20,500  	& 	N/A   		&   618.29   &  N/A   \\
        500 &  126,250	& 	N/A   		&   4,084.79 & 	N/A   \\
    \hline
    \hline
    \end{tabular}%
  \label{tab:RE}%
  \newline \footnotesize{For the case $p=200$ and $p=500$ the QMLE failed to converge. \\ The time complexity and RE is the average over $N=399$ simulations. \\ All simulations/estimations are done using a single core.}
\end{table}%

The (average) computation times and the relative efficiency for the two estimators are contained in table \ref{tab:RE}. For the larger systems, $p>20$, the computation time for QMLE is very big, on average 50 minutes for $p=50$ and 260 minutes for $p=100$, whereas STE remains feasible in all but the $p=500$ case, in which the computation time is roughly one hour. Considering the relative efficiency of the two estimators, the QMLE performs favourably for $p\leq10$, after which its performance deteriorate drastically compared to the STE. For portfolios larger than 10 assets, the STE is preferred. 

Here, estimations are initiated in $\vartheta_{init}=\vartheta_0-0.025$. However, one could argue that initiating the both estimation procedures in $\vartheta_{init}\neq\vartheta_{0}$ gives the joint QMLE a disadvantage, as it performs numerical optimization over a much larger parameter space. We therefore repeat the exercise, initiating in $\vartheta_{init}=\vartheta_0$. This yields almost identical results (available upon request) and leads to the same conclusion, namely that the STE is relatively more efficient than joint QMLE for systems larger than $p>10$ assets, and that it always has a lower computational complexity than joint QMLE.

\subsection{An application in risk management}% -- SVTE vs. VT-QMLE vs. QMLE} 
\label{sec:emp_ill_app}
We now turn our attention to the empirical performance of the STE of the $\lambda$-GARCH, and compare it to the joint QML estimator. The out-of-sample performance is assessed by considering the conditional $5\%$ Value-at-Risk (VaR) for five different medium-sized portfolios consisting of $p=25$ assets from the SP100 index.

\subsubsection{Methodology and data}
We consider the out-of-sample performance by considering the conditional $5\%$ value-at-risk at 1 and 5-day horizons for five different portfolios. The first of the five portfolios is equally weighted %, $w_i=1/p$ for $i=1,\hdots,p$, 
while the weights of the remaining portfolios are drawn randomly such that the second and third portfolios are long-only, with the third portfolio $50\%$ geared. The fourth and fifth portfolios are long-short portfolios. The constituents of the portfolios are drawn randomly from the SP100 index and can, along with their weighting, be found in Appendix \ref{appendix:empirical_ex}.

Each of the three estimators is fitted on a (rolling window) sample of $T=1200$ daily observations, with the initial sample starting on December $28$th 2010 and ending on December $29$th 2015. The out-of-sample consists of 3 years of data from December $30$th 2015 to December $31$st 2018, leading to $\tau=756$ out-of-sample observations for the $1$-day forecast and $\tau=189$ observations for the $5$-day (non-overlapping) forecasts. The out-of-sample forecasts are computed using a filtered historical simulation in which we draw innovations $iid$ with replacement from the standardized residuals, $\hat Z_t$, see e.g. \citet{christoffersen2009}.  

Recall that the conditional VaR at risk level $\alpha$ for the $h$-period return of portfolio $i$, denoted VaR$^i_{t,h}(\alpha)$ is defined as,
\begin{align}
	P_t(R^i_{t+h|t}<-\text{VaR}_{t,h}^i(\alpha)) = \alpha,
\end{align}
where $P_t$ is the conditional distribution of the ex ante $h$-period return of portfolio $i$, $R_{t+h|h}^i$% with $p\times1$ vector of weights $\mathcal{W}_i$, $R^i_{t+h|t}=\mathcal{W}_i'E[X_{t+h}|\mathcal{F}_{t}]$, conditional on the time $t$ information set, $\mathcal{F}_{t}$, which is the $\sigma$-field generated by $\{X_s\}_{s<t}$
. Define the (unconditional) ``hit'' variable for portfolio $i$ as,
\begin{align}
	I^i_{t} = \textbf{1}\{R^i_{t+h|t}>-\text{VaR}^i_{t,h}(\alpha) \},
\end{align}
such that the unconditional coverage for portfolio $i$ is $\pi_i = \frac{1}{\tau}\sum_{n=1}^\tau I^i_{n}$. Similarly, we define the conditional hit variable as $I^i_{t|t-1} = \textbf{1}\{I_{t}^i=1 | I_{t-1}^i=1\}$, denoting two hits in a row.

When assessing the adequacy of the VaR forecasts we consider the three likelihood ratio (LR) tests proposed by \citet{christoffersen1998}. The first LR test examines the hypothesis that the unconditional coverage is correct, $E[I_t^i]=\alpha$, but fails to account for potential clustering in the VaR hits. This is rectified by the second test, in which $I_{t|t-1}^i$ follows a two-state Markov chain, and we test the hypothesis of independence between hits. However, this test does not test for correct coverage, and as a consequence, we also consider the third test of correct conditional coverage, which lets $I_{t|t-1}^i$ follow the two-state Markov chain, and tests it against the null of independence between hits and correct coverage. The tests are denoted $LR_{uc}$, $LR_{ind}$ and $LR_{cc}$ respectively.

\begin{table}[htbp]
  \centering
  \caption{Empirical exercise: $95\%$ VaR at 1 and 5-day forecast horizons}
    %\footnotesize
    {
    \begin{tabular*}{\textwidth}{@{\extracolsep{\fill}}llccccc}
    %\begin{tabular}{llccccc}
    \hline
    \hline
          &       & \multicolumn{2}{c}{1 day} &       & \multicolumn{2}{c}{5 day} \\
\cline{3-7}          &       & STE  & QMLE  &       & STE  & QMLE \\ \hline
    $P_1$ & $\hat\pi_1$ & 0.044 & 0.038 &       & 0.048 & 0.048 \\
          & $LR_{uc}$ & 0.413 & 0.126 &       & 0.880 & 0.880 \\
          & $LR_{ind}$ & 0.013 & 0.910 &       & 0.427 & 0.056 \\
          & $LR_{cc}$ & 0.033 & 0.309 &       & 0.721 & 0.160 \\
    \hline
    $P_2$ & $\hat\pi_2$ & 0.046 & 0.040 &       & 0.053 & 0.058 \\
          & $LR_{uc}$ & 0.636 & 0.178 &       & 0.856 & 0.614 \\
          & $LR_{ind}$ & 0.093 & 0.478 &       & 0.090 & 0.136 \\
          & $LR_{cc}$ & 0.218 & 0.313 &       & 0.234 & 0.290 \\
    \hline
    $P_3$ & $\hat\pi_3$ & 0.042 & 0.038 &       & 0.053 & 0.058 \\
          & $LR_{uc}$ & 0.321 & 0.126 &       & 0.856 & 0.614 \\
          & $LR_{ind}$ & 0.050 & 0.910 &       & 0.537 & 0.136 \\
          & $LR_{cc}$ & 0.089 & 0.309 &       & 0.813 & 0.290 \\
    \hline
    $P_4$ & $\hat\pi_4$ & 0.065 & 0.052 &       & 0.074 & 0.069 \\
          & $LR_{uc}$ & 0.073 & 0.842 &       & 0.155 & 0.261 \\
          & $LR_{ind}$ & 0.307 & 0.409 &       & 0.078 & 0.269 \\
          & $LR_{cc}$ & 0.119 & 0.697 &       & 0.077 & 0.288 \\
    \hline
    $P_5$ & $\hat\pi_5$ & 0.065 & 0.050 &       & 0.058 & 0.053 \\
          & $LR_{uc}$ & 0.073 & 0.973 &       & 0.614 & 0.856 \\
          & $LR_{ind}$ & 0.453 & 0.449 &       & 0.243 & 0.290 \\
          & $LR_{cc}$ & 0.152 & 0.751 &       & 0.446 & 0.562 \\
    \hline
    \hline
          &       &       &       &       &       &  \\
    \hline
    \hline
          &       & \multicolumn{2}{c}{STE} &       & \multicolumn{2}{c}{QMLE} \\
    \hline
    \multicolumn{2}{c}{Time complexity} & \multicolumn{2}{c}{9.1} &       & \multicolumn{2}{c}{526.5} \\
    \hline
    \hline
    \end{tabular*}%
    %\end{tabular}%
    }
  \label{tab:VaR}% 
  \vspace{0.15cm} \newline
	\footnotesize{$P_i$ refers to the $i$'th portfolio, with $\hat\pi_i$ being the unconditional hit ratio $i=1,...,5$. $LR_{uc}$, $LR_{ind}$ and $LR_{cc}$ are the asymptotic p-values for the LR test for unconditional coverage, independence and conditional coverage respectively. The time complexity is given in seconds and is computed using a single core for one out-of-sample iteration.} 
\end{table}%

\subsubsection{Out-of-sample results}
The results of the out-of-sample exercise is given in table \ref{tab:VaR}. Importantly, the STE procedure is roughly $57$ times faster than the QMLE. We note that the estimated $\lambda-$GARCH (on average) has finite second order moments but not fourth order moments. Intuitively, this means that both estimators are consistent, but only the QML estimator has a limiting Gaussian distribution. 

The two estimation methods have a similar performance based unconditional coverage and the LR-tests: In general, the unconditional coverage is slightly different from the hypothesized $5\%$ and most of the LR-tests do not reject. Similar results are found for the 1 and 5-day $1\%$ VaR (not reported here). The rejected LR-tests relate to the equally weighted  portfolio $P_1$. 

In general, the VaR estimates produced by the two estimation methods are similar, but not identical: Consider figure \ref{fig:VaR_pf4} which plots the estimated VaR for the two estimation methods along with the realized return of portfolio $P_4$. As shown, the VaR estimates are, for the majority of the sample, very similar, but the QML estimator sometimes produce more extreme VaR estimates than the STE. We note, however, that while the two VaR estimates at times differ, the unconditional and conditional hit sequences are almost identical, and based on the LR-test in table \ref{tab:VaR}, none of the estimation methods seem to dominate the other empirically. We therefore conclude that the estimation procedures seem to yield similar results, with the STE having the clear advantage that it is much faster in practice. 

\begin{figure}[h]
	\caption{1 day $95\%$ VaR for portfolio $P_4$}
	\centering
	\includegraphics[width=\textwidth]{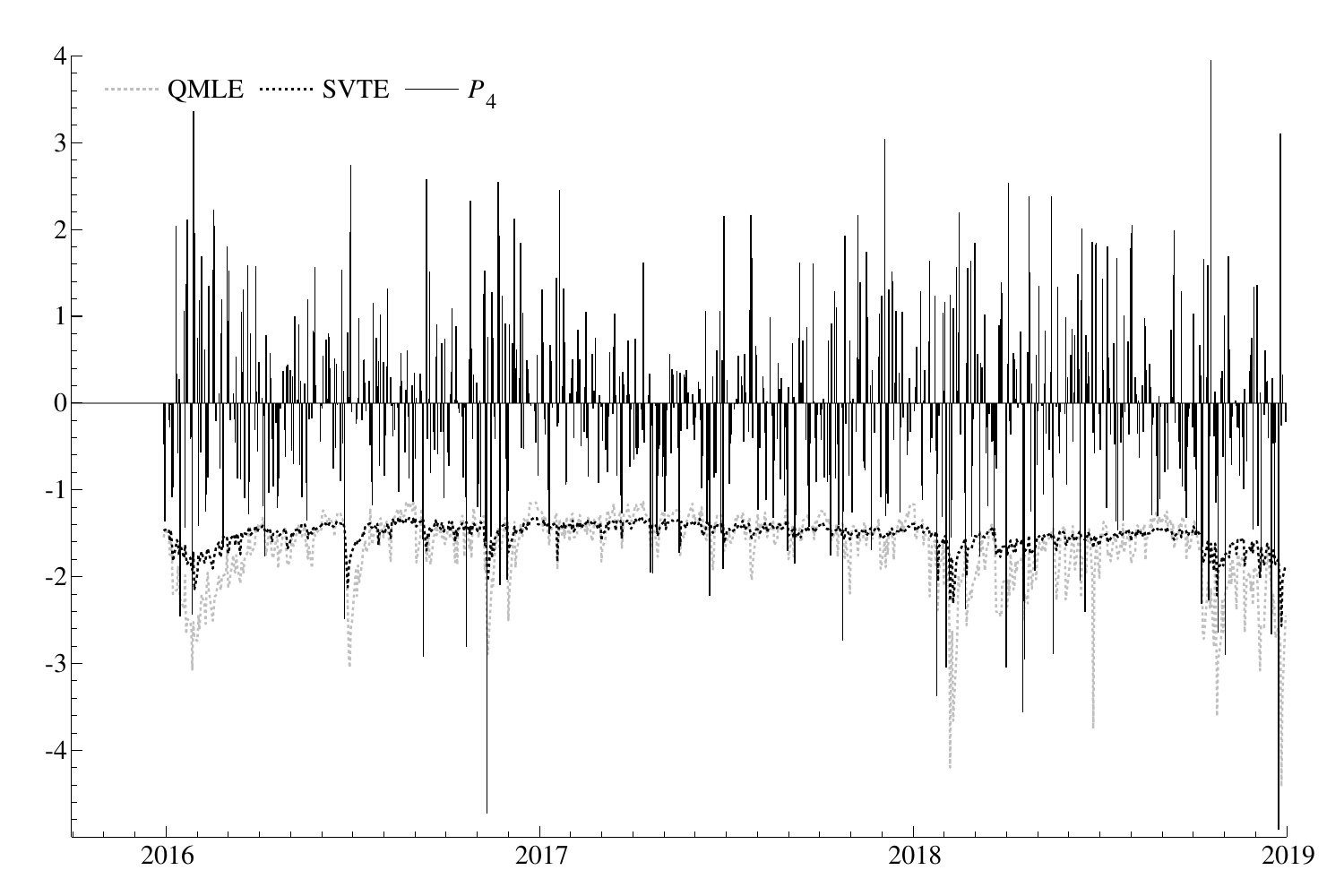}
	\\ \footnotesize{Note: STE and QMLE denote the estimated 1 day $5\%$ VaR, $P_4$ denotes the realized return.}
	\label{fig:VaR_pf4}
\end{figure}

\subsection{Brief summary of numerical exercises}
The simulation evidence in \ref{sec:relative_efficiency} indicates that not only is the STE relatively more efficient than QMLE in cross-sections of more than $p>10$ assets, it is also much more time efficient. This is verified in by the empirical study in Section \ref{sec:emp_ill_app}. One potential explanation is that the $\lambda-$GARCH is a non-linear function of the parameters in $\gamma$ through $Y_{t-1}$. By using a stepwise estimator, in which $\gamma$ is estimated using a closed form estimator, we mitigate the potential issues due to non-linearity, which seem to cause issues for large $p$ in the QML estimator. 

In regards to the asymptotic results in Theorem \ref{theorem:consistency}-\ref{theorem:asympnorm}, the simulation study in Appendix \ref{appendix:simulation} suggests that estimator is consistent in the case of finite second order moments of $X_t$. Furthermore, the simulations indicate that the asymptotic normality of the STE holds when $X_t$ has finite fourth moments. Hence, the moment requirements in Assumption \ref{assumption1} and \ref{assumption4} appear to be sufficient and necessary for consistency and asymptotic normality of the estimator. 

\section{Extensions and Concluding remarks}
\label{sec:conclusion}
We have derived asymptotic properties of the spectral targeting estimator (STE) for the $\lambda-$GARCH, an extended version of the  multivariate orthogonal GARCH (O-GARCH). The two-step estimator is consistent under finite second order moments, while it has a limiting Gaussian distribution when fourth order moments are finite. Simulations indicate that these moment conditions are sufficient and necessary. Moreover, we compare the empirical performance of the STE to that of the quasi maximum likelihood estimator (QMLE) for five portfolios of 25 assets. The STE dominates QMLE in terms of computational complexity, being up to 57 times faster in estimation, while both estimators produce similar out-of-sample forecasts. Finally, simulations indicate that in portfolios of more than 10 assets, the stepwise estimator is relatively more efficient than QMLE. The STE is therefore well suited for practitioners as it alleviates numerical problems and speeds up numerical optimization, while being easy to implement.

We note that while the STE delivered promising results in this exposition, the first step (sample) estimator may not be well-behaved when the ratio $p/T$ approaches one. This is discussed in e.g. \citet{ledoit2004,ledoit2012}, who derive shrinkage estimators for the sample covariance matrix, minimizing the estimation error. An extension could therefore consider the asymptotic analysis of a spectral targeting estimator where the first step estimator is based on shrinkage. Another extension would be to consider spectral targeting estimation with infinite fourth order moments, in a similar fashion to the exposition in \citet{Pedersen2014b} who consider the variance targeting estimator.

%\newpage
\printbibliography

%\bibliographystyle{unsrt}
%\bibliography{library}

%\newpage
\appendix

\section{Proofs}
\label{appendix:proofs}
%In the asymptotic analysis we assume that the data-generating process $\{X_t\}_{t\in\mathbb{N}}$ is strictly stationary, ergodic and has finite second order moments, $\rho(A_0+B_0)<1$, and is initiated in the infinite past. %In applications, the log-likelihood function is based on some fixed initial values, $X_0$ and $H_0$. We show, however, that the log-likelihood function with fixed initial values converge to the (unobserved) strictly stationary log-loglikelihood function for $T\rightarrow \infty$, and hence the choice of initial values are asymptotically irrelevant. 
Recall the log-likelihood function for the $i$'th equation,
\begin{align}
	L_T^{(i)}(\gamma,\kappa^{(i)}) & = \frac{1}{T}\sum_{t=1}^T l_t^{(i)}(\gamma,\kappa^{(i)}), \label{eq:loglik_app0} \\
	l_t^{(i)}(\gamma,\kappa^{(i)}) & = \log(\lambda_{i,t}(\gamma,\kappa^{(i)}))+\frac{y_{i,t}^2(\gamma)}{\lambda_{i,t}(\gamma,\kappa^{(i)})},
	\label{eq:loglik_app1}
\end{align}
which has first and second order derivatives,
\begin{align}
	\frac{\partial l_t^{(i)}(\theta^{(i)})}{\partial\theta^{(i)}_n} & = 
		\left(1- \frac{y_{i,t}^2(\theta^{(i)})}{\lambda_{i,t}(\theta^{(i)})}\right)\frac{1}{\lambda_{i,t}(\theta^{(i)})}\frac{\partial \lambda_{i,t}(\theta^{(i)})}{\partial \theta^{(i)}_n} + 2\frac{y_{i,t}(\theta^{(i)})}{\lambda_{i,t}(\theta^{(i)})}\frac{\partial y_{i,t}(\theta^{(i)})}{\partial \theta^{(i)}_n}, \label{eq:loglik_app_1der}
		\\
	\frac{\partial^2 l_t^{(i)}(\theta^{(i)})}{\partial\theta^{(i)}_n\partial\theta^{(i)}_m} & =
		\left(1- \frac{y_{i,t}^2(\theta^{(i)})}{\lambda_{i,t}(\theta^{(i)})}\right)\frac{1}{\lambda_{i,t}(\theta^{(i)})}\frac{\partial^2 \lambda_{i,t}(\theta^{(i)})}{\partial \theta^{(i)}_n\partial \theta^{(i)}_m}		\notag \\
		& + \left(2\frac{y_{i,t}^2(\theta^{(i)})}{\lambda_{i,t}(\theta^{(i)})}-1\right)\frac{1}{\lambda_{i,t}^2(\theta^{(i)})}\frac{\partial \lambda_{i,t}(\theta^{(i)})}{\partial \theta^{(i)}_n}\frac{\partial \lambda_{i,t}(\theta^{(i)})}{\partial \theta^{(i)}_m} \notag \\
		&+ 2\left( \frac{\partial^2 y_{i,t}(\theta^{(i)})}{\partial \theta^{(i)}_n \partial \theta^{(i)}_m}y_{i,t}(\theta^{(i)})	 + \frac{\partial y_{i,t}(\theta^{(i)})}{\partial \theta^{(i)}_n}\frac{\partial y_{i,t}(\theta^{(i)})}{ \partial \theta^{(i)}_m }
		\right)\frac{1}{\lambda_{i,t}(\theta^{(i)})} \notag \\
		& - 2\left(\frac{\partial \lambda_{i,t}(\theta^{(i)})}{\partial\theta^{(i)}_n}\frac{\partial y_{i,t}(\theta^{(i)})}{\partial\theta^{(i)}_m}+\frac{\partial \lambda_{i,t}(\theta^{(i)})}{\partial\theta^{(i)}_m}\frac{\partial y_{i,t}(\theta^{(i)})}{\partial\theta^{(i)}_n}
		\right)\frac{y_{i,t}(\theta^{(i)})}{\lambda_{i,t}^2(\theta^{(i)})}
		, \label{eq:loglik_app_2der}
\end{align}
for $n,m=1,\hdots, p(p+1)+1.$

Throughout the proofs, we let $L_{t}^{(i)}(\theta^{(i)})$ ($l_{t}^{(i)}(\theta^{(i)})$) denote the log-likelihood function (-contribution) initiated in the infinite past, and we let $L_{t,h}^{(i)}(\theta^{(i)})$ ($l_{t,h}^{(i)}(\theta^{(i)})$) denote the log-likelihood function (-contribution) initiated in a fixed $X_0$ and $H_0$ (with $H_0$ positive definite),
\begin{align}
	L_{T,h}^{(i)}(\gamma,\kappa^{(i)}) & = \frac{1}{T}\sum_{t=1}^T l_{t,h}^{(i)}(\gamma,\kappa^{(i)}), \label{eq:loglik_app2} \\
	l_{t,h}^{(i)}(\gamma,\kappa^{(i)}) & = \log(\lambda_{i,t,h}(\gamma,\kappa^{(i)}))+\frac{y_{i,t}^2(\gamma)}{\lambda_{i,t,h}(\gamma,\kappa^{(i)})},
	\label{eq:loglik_app3}
\end{align}
where $\lambda_{i,t}(\theta^{(i)})$ and $\lambda_{i,t,h}(\theta^{(i)})$ is defined analogously. Because $\lambda_{i,t}(\theta^{(i)})$ and $\lambda_{i,t,h}(\theta^{(i)})$ are defined for the same strictly stationary and ergodic sequence, $\{X_t\}_{t\in\mathbb{N}}$, we may write,
\begin{align}
		\lambda_{i,t}(\theta^{(i)})-\lambda_{i,t,h}(\theta^{(i)}) = b_i^t(\lambda_{i,0}(\theta^{(i)})-\lambda_{i,0,h}(\theta^{(i)})),
		\label{eq:lambda_app_diff}
\end{align}
for $t\geq1$

The structure of the main proofs and the accompanying lemmata follow that of \citet{Pedersen2014} (proof of Theorems 4.1-4.2 and Lemmata B.1-B.11). In order to make the proofs readable, most steps rely on lemmata stated and proved in Appendix \ref{appendix:lemmas}. In the following, we let the letters $K$ and $\phi$ denote generic constants, whose value can vary along the text, but always satisfy $K>0$ and $0<\phi<1$. Furthermore, let $H_{0,t}:=H_t(\theta_0^{(i)})$, $V_0:=V(\theta_0^{(i)})$ and $\Lambda_{0,t}:=\Lambda_{t}(\theta_0^{(i)})$. 

\subsection{Proof of consistency}
Initially, observe that by the ergodic theorem (Theorem 20.3 of \citet{jacod2012}), along with Assumption \ref{assumption1} the sample estimator is strongly consistent, 
%\begin{align}
	$\hat H\overset{a.s.}{\rightarrow} H_0$, % \ \ \ \ \ \ \text{a.s.},$
	%\label{eq:consistency_gamma}
%\end{align}
 for $T\rightarrow \infty$. Since $\lambda_0$ is assumed to be simple, we may use the continuous mapping theorem (Theorem 17.5 of \citet{jacod2012}) to establish strong consistency of the first stage estimation,
\begin{align}
	\hat\lambda&\rightarrow \lambda_0 \ \ \ \ \ \ \text{a.s.}\label{eq:consistency_gamma0} \\
	\hat\upsilon&\rightarrow \upsilon_0 \ \ \ \ \ \ \text{a.s.} \label{eq:consistency_gamma}
\end{align}
We now show that $\hat\kappa^{(i)}$ is consistent. The proof follows that of Theorem 4.1 in \citet{Pedersen2014}. 

For any $\epsilon>0$, it holds almost surely for large T,
\begin{align*}
	 E[l_t^{(i)}(\gamma_0,\hat\kappa^{(i)})]&<L_T^{(i)}(\gamma_0,\hat\kappa^{(i)})+\frac{\epsilon}{5} \ \ \ \ \ \ \ \ \ \ \ \ \text{By Lemma \ref{lemma:con2}, } \\
	 L_T^{(i)}(\gamma_0,\hat\kappa^{(i)})&<L_{T,h}^{(i)}(\hat\gamma,\hat\kappa^{(i)})+\frac{\epsilon}{5} \ \ \ \ \ \ \ \ \ \ \ \ \text{By Lemma \ref{lemma:con4}, } \\ 
	 L_{T,h}^{(i)}(\hat\gamma,\hat\kappa^{(i)})&<L_{T,h}^{(i)}(\hat\gamma,\kappa_0^{(i)})+\frac{\epsilon}{5} \ \ \ \ \ \ \ \ \ \ \ \ \text{By \eqref{eq:QMLE}, } \\ 
	 L_{T,h}^{(i)}(\hat\gamma,\hat\kappa^{(i)})&<L_T^{(i)}(\gamma_0,\kappa_0^{(i)})+\frac{\epsilon}{5} \ \ \ \ \ \ \ \ \ \ \ \ \text{By Lemma \ref{lemma:con4}, } \\ 
	 L^{(i)}_T(\gamma_0,\kappa_0^{(i)})&<E[l_t^{(i)}(\gamma_0,\kappa_0^{(i)})]+\frac{\epsilon}{5} \ \ \ \ \ \ \ \ \ \text{By Lemma \ref{lemma:con2} }.
\end{align*}
That is, for any $\epsilon>0$,
\begin{align*}
	E[l_t^{(i)}(\gamma_0,\hat\kappa^{(i)})]<E[l_t^{(i)}(\gamma_0,\kappa_0^{(i)})]+\epsilon,
\end{align*}
and by Lemma \ref{lemma:con3} along with standard arguments for two-step estimators \citep{Newey1994}, it follows that for $T\rightarrow\infty$, $\hat\kappa^{(i)}\overset{a.s.}{\rightarrow}\kappa^{(i)}_0$. By Hence, the two-step estimator is strongly consistent, $\hat\theta^{(i)}\overset{a.s.}{\rightarrow}\theta^{(i)}$.
\qed
\subsection{Proof of asymptotic normality}
Compared to asymptotic theory for the joint QMLE of multivariate GARCH models additional difficulties arise from the fact that STE is a multi-step estimator. Conversely, the proof is simplified by the additional assumption of $E||X_t||^4<\infty$ and the fact that we treat individual $\{y_{i,t}\}_{t\in\mathbb{N}}$ separately.

The proof of asymptotic normality is based on an application of the mean-value theorem on the optimality condition of the score vector, $\theta^{(i)}=\theta_0^{(i)}$ along with Assumption \ref{assumption5} and \eqref{eq:QMLE},
\begin{align}
0_{p+1\times1} &  %= \frac{1}{{T}}\sum_{t=1}^T\frac{\partial l_t^{(i)}(\hat\theta^{(i)})}{\partial \kappa^{(i)}} 
     = \frac{\partial L_{T,h}^{(i)}(\theta^{(i)}_0)}{\partial \kappa^{(i)}} + \left(\frac{\partial^2L_{T,h}^{(i)}(\tilde\theta^{(i)})}{\partial\kappa^{(i)}\partial\theta^{(i)'}}\right)(\hat\theta^{(i)}-\theta^{(i)}_0) \notag \\
  & = \frac{\partial L_{T,h}^{(i)}(\theta_0^{(i)})}{\partial \kappa^{(i)}} + J^{(i)}_{T,h}(\tilde\theta^{(i)})(\hat\kappa^{(i)}-\kappa_0^{(i)})+K^{(i)}_{T,h}(\tilde\theta^{(i)})(\hat\gamma-\gamma_0),
\end{align}
where
\begin{align*}
	\frac{\partial L_{t,h}^{(i)}(\theta^{(i)}_0)}{\partial \kappa^{(i)}}=\frac{\partial L_{t,h}^{(i)}(\theta^{(i)})}{\partial \kappa^{(i)}}\Bigg|_{\theta^{(i)}=\theta^{(i)}_0},  
	\  
	J_{T,h}^{(i)}(\tilde\theta^{(i)}) = \frac{\partial^2 L_{T,h}^{(i)}(\theta^{(i)})}{\partial\kappa^{(i)}\partial\kappa^{(i)'}}\Bigg|_{\theta^{(i)}=\tilde\theta^{(i)}}, 
	\  
	K_{T,h}^{(i)}(\tilde\theta^{(i)}) =  \frac{\partial^2 L_{T,h}^{(i)}(\theta^{(i)})}{\partial\kappa^{(i)}\partial\gamma'}\Bigg|_{\theta^{(i)}=\tilde\theta^{(i)}}.
\end{align*}
Here $\kappa^{(i)}=[a_{i1},\hdots,a_{ip}, b_i]'$, $\gamma=[\lambda',\upsilon']'$, and $\tilde\theta^{(i)}$ is on the line between $\theta_0^{(i)}$ and $\hat\theta^{(i)}$.

$J^{(i)}_T$ is finite and invertible with probability approaching one (Lemma \ref{lemma:norm2} and \ref{lemma:norm4}) and $\hat\theta^{(i)}\overset{a.s.}{\rightarrow}\theta^{(i)}$ (Theorem \ref{theorem:consistency}). Hence, by Lemma \ref{lemma:norm5}, 
\begin{align}
	&\sqrt{T}\left(\hat\kappa^{(i)}-\kappa_0^{(i)}\right) = -\left(J^{(i)}_T(\tilde\theta^{(i)})\right)^{-1}\sqrt{T}\frac{\partial L_t^{(i)}(\theta_0^{(i)})}{\partial \kappa^{(i)}} -\left(J^{(i)}_T(\tilde\theta^{(i)})\right)^{-1}K^{(i)}_T(\tilde\theta^{(i)})\sqrt{T}\left(\hat\gamma-\gamma_0\right)+o_P(1),
\end{align}
which we rewrite for the (joint) parameter vector of equation $i$,
\begin{align}
	\sqrt{T}
	\begin{pmatrix}
		\hat\gamma-\gamma_0 \\
		\hat\kappa^{(i)}-\kappa^{(i)}_0
	\end{pmatrix}
	=
	\begin{pmatrix}
		I_{p(p+1)} & 0_{p(p+1)\times p+1} \\
		-(J_T^{(i)}(\tilde\theta^{(i)}))^{-1}K_T^{(i)}(\tilde\theta^{(i)}) & -(J_T^{(i)}(\tilde\theta^{(i)}))^{-1}
	\end{pmatrix}
	\begin{pmatrix}
		\sqrt{T}(\hat\gamma-\gamma_0) \\
		\sqrt{T}\frac{\partial L_t^{(i)}(\theta_0^{(i)})}{\partial\kappa^{(i)}}
	\end{pmatrix}
	+o_P(1).
\end{align}
The asymptotic normality then follows from Lemma \ref{lemma:norm1} together with Slutsky's theorem,
\begin{align}
	\sqrt{T}\left(\hat\theta^{(i)}-\theta^{(i)}_0\right)\overset{D}{\rightarrow}N(0,\Sigma_0^{(i)}),
\end{align}
with
\begin{align}
	\underset{e\times e}{\Sigma_0^{(i)}} & = 	
	\begin{pmatrix}
		I_{p(p+1)} & 0_{p(p+1)\times p+1} \\
		-(J_0^{(i)})^{-1}K_0^{(i)} & -(J_0^{(i)})^{-1}
	\end{pmatrix}
	\Omega_0^{(i)}
	\begin{pmatrix}
		I_{p(p+1)} & -(J_0^{(i)})^{-1}(K_0^{(i)})' \\
		0_{p+1\times p(p+1)} & -(J_0^{(i)})^{-1}
	\end{pmatrix},
\end{align}
where $\Omega_0^{(i)}$ is defined in Lemma \ref{lemma:norm1}, and $J_{T}^{(i)}(\tilde\theta^{(i)})\overset{a.s.}{\rightarrow} J_0^{(i)}$, $K_{T}^{(i)}(\tilde\theta^{(i)})\overset{a.s.}{\rightarrow} K_0^{(i)}$ by Lemma \ref{lemma:norm3}, with
\begin{align}
	\underset{p+1\times p+1}{J_0^{(i)}} = E\left[\frac{\partial^2 l_t^{(i)}(\theta^{(i)})}{\partial\kappa^{(i)}\partial\kappa^{(i)'}}\middle|_{\theta^{(i)}=\theta^{(i)}_0}\right], \ \ \ \ \ 	\underset{p+1\times p^2+p}{K_0^{(i)}} = E\left[\frac{\partial^2 l_t^{(i)}(\theta^{(i)})}{\partial\kappa^{(i)}\partial\gamma'}\middle|_{\theta^{(i)}=\theta^{(i)}_0}\right].
	\label{eq:appendix_asympvar}
\end{align}
\qed

\subsection{Proof of Corollary \ref{corollary:dist_w_l}}
The asymptotic distribution of the constant term in the $i$th conditional eigenvalue is found using the delta method, for which we need the partial derivative of $w_i$ with respect to the parameter vector $\theta^{(i)}$,
\begin{align}
	\underset{e\times1}{\Phi_0'} = \frac{\partial w_i}{\partial \theta^{(i)}}\Bigg\vert_{\theta^{(i)}=\theta^{(i)}_0} = 
	\begin{pmatrix}
		1-b_{0,i}\mathbb{I}\{\lambda_{0,i}\}-\sum_{j=1}^p a_{0,ij} \\
		0_{p^2\times 1} \\
		-\lambda_0 \\
		-\lambda_{0,i}
	\end{pmatrix}
\end{align}
where $\mathbb{I}\{\lambda_i\}$ is a vector of zeros, apart from a $1$ in the $i$'th row. Hence, the asymptotic distribution of $w_i$ is,
\begin{align}
	\sqrt{T}(\hat w_i-w_{0,i}) & \overset{D}{\rightarrow} N(0,\Phi_0 \Sigma^{(i)}_0 \Phi_0')
\end{align}
\qed

\section{Lemmata} 
\label{appendix:lemmas}
\subsection{Lemmata for the proof of consistency}
\begin{lemma}[Finite expectation of likelihood contributions]
	Under Assumptions \ref{assumption1}-\ref{assumption3}, 
	\begin{align*}
		E\left[\sup_{\theta^{(i)}\in\Theta^{(i)}}\left|l_t^{(i)}(\gamma,\kappa^{(i)})\right|\right]\leq K,
	\end{align*}
	where $l_t^{(i)}(\gamma,\kappa^{(i)})$ is defined in \eqref{eq:loglik_app1}.
	\label{lemma:con1}
\end{lemma}
\begin{proof}
	Notice that the $i'$th conditional eigenvalue may be rewritten as an ARCH$(\infty)$ process,
	\begin{align*}
		\lambda_{i,t}(\theta^{(i)}) & = w_i+\sum_{j=1}^pa_{ij}y_{j,t-1}^2(\theta^{(i)})+b_i\lambda_{i,t-1}(\theta^{(i)}) 
									  = \sum_{l=0}^\infty b_i^l\left(w_i+\sum_{j=1}^pa_{ij}y_{j,t-l-1}^2(\theta^{(i)})  \right).
	\end{align*}
	Using $\rho(B)<1$ (By Assumption \ref{assumption1} and Lemma 4.1 of \citet{Ling2003}), along with Theorem 9.2 of \citet{jacod2012},
	\begin{align*}
		E\left[\sup_{\theta^{(i)}\in\Theta^{(i)}}\left|\log(\lambda_{i,t}(\theta^{(i)}))\right|\right] \leq E\left[\sup_{\theta^{(i)}\in\Theta^{(i)}}\left|\lambda_{i,t}(\theta^{(i)})\right|\right]\leq K \sum_{t=1}^\infty \phi^t(1+E||X_{t}||^2)\leq K.
	\end{align*} 
	Furthermore,
	\begin{align}
	E\left[\sup_{\theta^{(i)}\in\Theta^{(i)}}\left|\frac{1}{\lambda_{i,t}(\theta^{(i)})}\right|\right] &\leq\sup_{\theta^{(i)}\in\Theta^{(i)}}\left|\frac{1}{w_i}\right|\leq K,
	\label{eq:l_invers}	 \\
	E\left[\sup_{\theta^{(i)}\in\Theta^{(i)}}\left|y_{i,t}^2\right|\right] & \leq E\left[\sup_{\theta^{(i)}\in\Theta^{(i)}}\left|(V_i(\theta^{(i)})'X_{t})^2\right|\right] \leq KE\left[||X_{t})||^2\right] \leq K.
	\end{align}
	This, along with the triangle inequality, means that the log-likelihood contribution for the $i'$th rotated return is then bounded by a constant by Assumption \ref{assumption1},
	\begin{align*}
		E\left[\sup_{\theta^{(i)}\in\Theta^{(i)}}\left|l_t^{(i)}(\gamma,\kappa^{(i)})\right|\right] %&\leq E\left[\sup_{\theta^{(i)}\in\Theta^{(i)}}\left| \log\lambda_{i,t}(\theta^{(i)})\right|\right]+E\left[\sup_{\theta^{(i)}\in\Theta^{(i)}}\left|\frac{y_{i,t}^2(\theta^{(i)})}{\lambda_{i,t}(\theta^{(i)})} \right|\right] \\
		& \leq
		E\left[\sup_{\theta^{(i)}\in\Theta^{(i)}}\left|
		\lambda_{i,t}(\theta^{(i)})\right|\right]+E\left[\sup_{\theta^{(i)}\in\Theta^{(i)}}\left|\frac{y_{i,t}^2(\theta^{(i)})}{\lambda_{i,t}(\theta^{(i)})}
		\right|\right]
		\leq K.
	\end{align*}
\end{proof}

\begin{lemma}[Uniform convergence of likelihood function]
	Under Assumptions \ref{assumption1}-\ref{assumption3}, 
	\begin{align*}
		\sup_{\theta^{(i)}\in\Theta^{(i)}}\left|
		L_T^{(i)}(\gamma,\kappa^{(i)})-E[l_t^{(i)}(\gamma,\kappa^{(i)})]
		\right|\overset{a.s.}{\rightarrow}0,
	\end{align*}
	where $L_T^{(i)}(\theta^{(i)})$ is the log-likelihood function for the $i$th rotated return and $l_t^{(i)}(\theta^{(i)})$ is the log-likelihood contribution for the $i$th rotated return at time $t$, stated in \eqref{eq:loglik_app0} and \eqref{eq:loglik_app1}.
	\label{lemma:con2}
\end{lemma}
\begin{proof}
	Follows from Lemma \ref{lemma:con1} and the uniform law of large numbers (Theorem A.2.2. of \citet{White1994})
\end{proof}

\begin{lemma}[Likelihood uniquely minimized]
	Under Assumptions \ref{assumption1}-\ref{assumption3}, 
	\begin{align*}
		E|l_t^{(i)}(\gamma_0,\kappa_0^{(i)})|<\infty,
	\end{align*}
	and if $\kappa^{(i)}\neq\kappa_0^{(i)}$,
	\begin{align*}
		E[l_t^{(i)}(\gamma_0,\kappa^{(i)})]>E[l_t^{(i)}(\gamma_0,\kappa^{(i)}_0)],
	\end{align*}
	where $l_t^{(i)}(\gamma,\kappa^{(i)})$ is defined in \eqref{eq:loglik_app1}.
	\label{lemma:con3}
\end{lemma}
\begin{proof}
	The first statement follows directly from Lemma \ref{lemma:con1}. 
	For the second statement, consider
	 \begin{align*}
	 	E[l_t^{(i)}(\gamma_0,\kappa^{(i)})]- & E[l_t^{(i)}(\gamma_0,\kappa^{(i)}_0)]  \\ 
	 = & E\left[\log\left(\frac{\lambda_{i,t}(\gamma_0,\kappa^{(i)})}{\lambda_{i,t}(\gamma_0,\kappa_0^{(i)})}\right)+y_{i,t}^2(\gamma_0)\left(\frac{1}{\lambda_{i,t}(\gamma_0,\kappa^{(i)})}-\frac{1}{\lambda_{i,t}(\gamma_0,\kappa_0^{(i)})}\right) \right] 
	\\
	= &  
	E\left[\lambda_{i,t}(\gamma_0,\kappa^{(i)}_0)\left(\frac{1}{\lambda_{i,t}(\gamma_0,\kappa^{(i)})}-\frac{1}{\lambda_{i,t}(\gamma_0,\kappa_0^{(i)})}\right)
	-\log\left(\frac{\lambda_{i,t}(\gamma_0,\kappa_0^{(i)})}{\lambda_{i,t}(\gamma_0,\kappa^{(i)}))}\right)
	 \right] 
	 \\
	= & 
	E\left[\frac{\lambda_{i,t}(\gamma_0,\kappa^{(i)}_0)}{\lambda_{i,t}(\gamma_0,\kappa^{(i)})}-1
	-\log\left(\frac{\lambda_{i,t}(\gamma_0,\kappa_0^{(i)})}{\lambda_{i,t}(\gamma_0,\kappa^{(i)})}\right)
	 \right]\geq 0,
	 \end{align*}
	 which uses $y_{i,t}(\gamma_0)=\lambda_{i,t}^{1/2}(\gamma_0,\kappa_0^{(i)})z_{i,t}$, where $z_{i,t}$ is $iid$ with unit variance. Notice that $\log x\leq x-1$ for $x>0$, and that $\log x=x-1$ only if $x=1$. This inequality is strict unless $\frac{\lambda_{i,t}(\gamma_0,\kappa^{(i)}_0)}{\lambda_{i,t}(\gamma_0,\kappa^{(i)})}=1$, which is ruled out, as it violates assumption \ref{assumption2} on identification.% This, however, violates Assumption \ref{assumption2} on identification of the model, and we conclude that the inequality holds strictly, and the log-likelihood function is uniquely minimized.
\end{proof}

\begin{lemma}[Asymptotic irrelevance of initial values]
	Under Assumptions \ref{assumption1}-\ref{assumption3}, 
	\begin{align*}
		\sup_{\kappa^{(i)}\in \mathcal{K}^{(i)}}\left| L_T^{(i)}(\gamma_0,\kappa^{(i)})-L_{T,h}^{(i)}(\hat\gamma,\kappa^{(i)})  \right|\overset{a.s.}{\rightarrow}0,
	\end{align*}
	where $L_t^{(i)}(\gamma,\kappa^{(i)})$ is defined in \eqref{eq:loglik_app0} and $L_{t,h}^{(i)}(\gamma,\kappa^{(i)})$ is defined in \eqref{eq:loglik_app2}.
	\label{lemma:con4}
\end{lemma}

\begin{proof}
	We want to show that the initial values are asymptotically irrelevant. 
	As in the proof of Theorem 4.1 in \citet{Pedersen2014b}, we use the triangle inequality as follows,
	\begin{align}
		&\sup_{\kappa^{(i)}\in \mathcal{K}^{(i)}}\left| L_T^{(i)}(\gamma_0,\kappa^{(i)})-L_{T,h}^{(i)}(\hat\gamma,\kappa^{(i)})  \right| \leq \notag \\
		&
		\sup_{\kappa^{(i)}\in \mathcal{K}^{(i)}}\left| L_T^{(i)}(\gamma_0,\kappa^{(i)})-L_T^{(i)}(\hat\gamma,\kappa^{(i)})\right|+\sup_{\kappa^{(i)}\in \mathcal{K}^{(i)}}\left| L_T^{(i)}(\hat\gamma,\kappa^{(i)})-L_{T,h}^{(i)}(\hat\gamma,\kappa^{(i)})\right|.
		\label{eq:lemmaB4_1}
	\end{align}
	In line with the aforementioned proof, we apply the mean-value theorem to the first term of \eqref{eq:lemmaB4_1}, 
	\begin{align}
		\sup_{\kappa^{(i)}\in\mathcal{K}^{(i)}}|L_T^{(i)}(\gamma_0,\kappa^{(i)})-L_T^{(i)}(\hat\gamma,\kappa^{(i)})|\leq \sum_{j=1}^{p(p+1)}(\hat\gamma_j-\gamma_{0,j})\frac{1}{T}\sum_{t=1}^T\sup_{\theta^{(i)}\in\tilde{\mathcal{L}} \times\mathcal{V}\times\mathcal{K}^{(i)}}\left|\frac{\partial l_t(\gamma,\kappa^{(i)}}{\partial \gamma_j}\right|,
	\end{align}
	where $\tilde{\mathcal{L}}$ is chosen to be a compact subset of $(0,\infty)^{p}$ such that $(I_p-A-B)\lambda \in (0,\infty)^p$ and bounded away from zero on $\tilde{\mathcal{L}}\times\mathcal{V}$, and such that $\lambda_0$ lies in the interior of $\tilde{\mathcal{L}}$. An expression for $\partial l_t(\gamma,\kappa^{(i)})/\partial \theta^{(i)}$ can be found in \eqref{eq:loglik_app_1der} along with derivatives of $\lambda_{i,t}(\theta^{(i)})$ in \eqref{eq:lambda_der_g}-\eqref{eq:lambda_der_b} in Lemma \ref{lemma:norm2}. Notice also that $E[\sup_{\gamma\in\mathcal{L}\times\mathcal{V}}||\partial y_{i,t}^2(\gamma)/\partial\gamma||]\leq K E||X_t||^2$. By Assumption \ref{assumption1} and $\rho(B)<1$ (follows Assumption \ref{assumption1}, see Lemma 4.1 in \citet{Ling2003}) $\sup_{\theta^{(i)}\in\tilde{\mathcal{L}}\times\mathcal{V}\times\mathcal{K}^{(i)}}|\partial \lambda_{i,t}(\theta^{(i)}/\partial \gamma_j|<\infty$. This, along with $\sup_{\theta^{(i)}\in\tilde{\mathcal{L}}\times\mathcal{V}\times\mathcal{K}^{(i)}}|1/\lambda_{i,t}(\kappa^{(i)})|\leq K$, ensure that $E\left[\sup_{\theta^{(i)}\in\tilde{\mathcal{L}}\times\mathcal{V}\times\mathcal{K}^{(i)}}\partial l_t^{(i)}(\gamma,\kappa^{(i)})/\partial\gamma_j  |\right]\leq \infty$. By the ergodic theorem and \eqref{eq:consistency_gamma0}-\eqref{eq:consistency_gamma} we find that, $\sup_{\kappa^{(i)}\in \mathcal{K}^{(i)}}\left| L_T^{(i)}(\gamma_0,\kappa^{(i)})-L_T^{(i)}(\hat\gamma,\kappa^{(i)})\right|\overset{a.s.}{\rightarrow}0$.
	
	Next, consider the second term of \eqref{eq:lemmaB4_1},
	\begin{align}
			&\sup_{\kappa^{(i)}\in \mathcal{K}^{(i)}}\left| L_T^{(i)}(\hat\gamma,\kappa^{(i)})-L_{T,h}^{(i)}(\hat\gamma,\kappa^{(i)})  \right| \leq
			K\frac{1}{T}\sum_{t=1}^T\sup_{\kappa^{(i)}\in  \mathcal{K}^{(i)}}\left|b_i^t(\lambda_{i,0}(\hat\gamma,\kappa^{(i)})-\lambda_{i,0,h}(\hat\gamma,\kappa^{(i)}))\right|+ \notag \\
			&
			\frac{1}{T}\sum_{t=1}^TK\sup_{\kappa^{(i)}\in  \mathcal{K}^{(i)}}\left| ||X_t||^2\left[\frac{1}{\lambda_{i,t,h}(\hat\gamma,\kappa^{(i)})}\left(\lambda_{i,t,h}(\hat\gamma,\kappa^{(i)})-\lambda_{i,t}(\hat\gamma,\kappa^{(i)})\right)\frac{1}{\lambda_{i,t}(\hat\gamma,\kappa^{(i)})}\right] \right| \leq \notag \\
			&  K\frac{1}{T}\sum_{t=1}^T\phi^t(1+||X_t||^2)
			\label{eq:lemmaB4_2}
	\end{align}
	by $\log x\leq x-1$ for $x\geq1$ along with \eqref{eq:lambda_app_diff},  \eqref{eq:l_invers} and $\sup_{\gamma\in \mathcal{H}}|V_i'X_t|\leq K||X_t||$. Additionally, we use that for any $j\in \mathbb{Z}$, $\sup_{\kappa^{(i)}\in\mathcal{K}^{(i)}}|b_i^j|\leq K\phi^j$ (\citet{francq2011} (p.611) or \citet{francq2004} (p.616)), along with \ref{eq:consistency_gamma} and the compactness of $\mathcal{K}^{(i)}$ for $T\rightarrow \infty$,
	\begin{align}
	K\frac{1}{T}\sum_{t=1}^T\sup_{\kappa^{(i)}\in  \mathcal{K}^{(i)}}\left|b_i^t(\lambda_{i,0}(\hat\gamma,\kappa^{(i)})-\lambda_{i,0,h}(\hat\gamma,\kappa^{(i)}))\right| \leq K\frac{1}{T}\sum_{t=1}^T\phi^t%+o(1)
	\text{\ \ \ \ \ a.s.} \label{eq:lemma4_1}
	\end{align}
	Hence,
	\begin{align*}
		\sup_{\kappa^{(i)}\in  \mathcal{K}^{(i)}}\left| L_T^{(i)}(\hat\gamma,\kappa^{(i)})-L_{T,h}^{(i)}(\hat\gamma,\kappa^{(i)})  \right| \leq  K\frac{1}{T}\sum_{t=1}^T\phi^t(1+||X_t||^2)%+o(1)
		 \ \ \ \ \ \ \text{a.s.}
	\end{align*}
	For any $\epsilon>0$, we use Markov's inequality and Assumption \ref{assumption1},
	\begin{align*}
		\sum_{t=1}^\infty P(\phi^t||X_t||^2>\epsilon)\leq \sum_{t=1}^\infty\frac{\phi^t E||X_t||^2}{\epsilon}<\infty,
	\end{align*}
	Next, by the Borel-Cantelli lemma, $\phi^t||X_t||^2\overset{a.s.}{\rightarrow}0$, and finally, by Cesaro's mean theorem
	\begin{align*}
		\frac{1}{T}\sum_{t=1}^T\phi^t||X_t||^2\overset{a.s.}{\rightarrow}0,
	\end{align*}
	we conclude that the initial values are asymptotically irrelevant for consistency of the estimator.
\end{proof}

\subsection{Lemmata for the proof of asymptotic normality}

\begin{lemma}[Rewriting the two-step estimator in vector form]
	Under Assumptions \ref{assumption1}-\ref{assumption5}, for $T\rightarrow \infty$, the two-step estimator can be written jointly as,
	\begin{align}
	&\sqrt{T}
	\begin{pmatrix}
		\hat\gamma-\gamma_0 \\
		\frac{\partial L_T^{(i)}(\theta_0)}{\partial\kappa}
	\end{pmatrix} 
	=
	%\sqrt{T}
	%\begin{pmatrix}
	%	\hat\lambda-\lambda_0 \\
	%	\hat\upsilon-\upsilon_0 \\
	%	\frac{\partial L_T^{(i)}(\theta_0)}{\partial\kappa}
	%\end{pmatrix} 
	%=
	\notag 
	  \\
	&
	\frac{1}{\sqrt{T}}\sum_{t=1}^T
	\begin{pmatrix}
		D(I_{p^2}-\tilde A_0-\tilde B_0)^{-1}(I_{p^2}-\tilde B_0)(\Lambda_{0,t}^{1/2})^{\otimes2}\text{vec}(Z_tZ_t'-I_p) \\
		V_{0,1}'\otimes(\lambda_{0,1}I_p-H_0)^+V_0^{\otimes2}(I_{p^2}-\tilde A_0-\tilde B_0)^{-1}(I_{p^2}-\tilde B_0)(\Lambda_{0,t}^{1/2})^{\otimes2}\text{vec}(Z_tZ_t'-I_p) \\
		\vdots \\
		V_{0,p}'\otimes(\lambda_{0,p}I_p-H_0)^+V_0^{\otimes2}(I_{p^2}-\tilde A_0-\tilde B_0)^{-1}(I_{p^2}-\tilde B_0)(\Lambda_{0,t}^{1/2})^{\otimes2}\text{vec}(Z_tZ_t'-I_p) \\
		-\frac{1}{\lambda_{i,t}(\theta_0^{(i)})}\frac{\partial \lambda_{i,t}(\theta_0^{(i)})}{\partial \kappa^{(i)}}(z_{i,t}^2-1)
	\end{pmatrix}
	+o_P(1). \label{eq:est_rewritten}
	\end{align}
	where $(\cdot)^+$ is the Moore-Penrose inverse, $D$ is $p\times p^2$ with all elements zero, apart from one element on each row, which is $d_{j,j+(j-1)p}=1$ for $j=1,\hdots,p$ and $\tilde A$ and $\tilde B$ are given in \eqref{eq:vec-lambda_AB}.
	\label{lemma:norm0}
\end{lemma}
\begin{proof}
%\color{red}
	In rewriting the estimator in vector form, we partly follow \citet{Pedersen2014} (Lemma B.8) and rewrite $\text{vec}(\hat H)-\text{vec}(H_0)$ in terms of the GARCH parameters. The remainder of the proof is distinctly different, as we have to recast the vector of dynamic eigenvalues, $\lambda_t$, in a BEKK$(p^2,1,1,1)$ parameterization, and state the first step estimator in terms of $\hat\gamma-\gamma_0$ rather than $\text{vec}(\hat H)-\text{vec}(H_0)$.

	First, consider the moment estimator of the unconditional covariance matrix,
	\begin{align*}
	 	\hat H = \frac{1}{T}\sum_{t=1}^TX_tX_t' =  \frac{1}{T}\sum_{t=1}^T H_{0,t}^{1/2}Z_tZ_t'(H_{0,t}^{1/2})'= \frac{1}{T}\sum_{t=1}^T V_0\Lambda_{0,t}^{1/2}Z_tZ_t'\Lambda_{0,t}^{1/2}V_0'.
	\end{align*}
	Recall that $Y_{0,t}=V'_0X_t$ and define $\hat\Lambda = \frac{1}{T}\sum_{t=1}^TY_{0,t}Y_{0,t}'=\frac{1}{T}\sum_{t=1}^T\Lambda_{0,t}^{1/2}Z_tZ_t'\Lambda_{0,t}^{1/2}$, such that,
	\begin{align}
		\text{vec}(\hat\Lambda) = (V_0')^{\otimes2}\text{vec}\left(\hat H\right) & = 	\text{vec}\left(\frac{1}{T}\sum_{t=1}^T \Lambda_{0,t}^{1/2}Z_tZ_t'\Lambda_{0,t}^{1/2}\right) \notag\\
	 	& = \left( \frac{1}{T}\sum_{t=1}^T(\Lambda_{0,t}^{1/2})^{\otimes2}\text{vec}(Z_tZ_t'-I_p) + \text{vec}\left( \frac{1}{T}\sum_{t=1}^T\Lambda_{0,t}\right)\right).
	 	\label{eq:mom_rewrite}
	\end{align}
	Next, we need to rewrite the conditional eigenvalues in a ``vec''-reparameterization. That is, we first write the dynamics of $\Lambda_{0,t}$ to be nested in the BEKK-GARCH, and then we apply the vec-operator to obtain the vec-parameterization of the conditional eigenvalues. Hence,
	\begin{align}
		\Lambda_{0,t}= C +\sum_{i=1}^{p^2}A_{i,0}Y_{0,t-1}Y_{0,t-1}'A_{i,0}'+B_{1,0}\Lambda_{0,t-1}B_{1,0}',
		\label{eq:BEKK-lambda}
	\end{align}
	with $C=(I_p-\sum_{i=1}^{p^2}A_{i,0}\Lambda_0A_{i,0}'+B_{1,0}\Lambda_{0}B_{1,0}')$, and $A_i$ are restricted parameter matrices, e.g. for the bivariate case,
	\begin{align}
		A_1 = \begin{pmatrix}
			\sqrt{a_{11}} & 0 \\ 0 & 0 
		\end{pmatrix},  \ \ 
		A_2 = \begin{pmatrix}
			0 & \sqrt{a_{12}} \\ 0 & 0 
		\end{pmatrix}, \ \ 
		A_3 = \begin{pmatrix}
			0 & 0 \\ \sqrt{a_{21}} & 0 
		\end{pmatrix}, \ \ 
		A_4 = \begin{pmatrix}
			0 & 0 \\ 0 & \sqrt{a_{22}}
		\end{pmatrix}
		 \label{eq:BEKK-lambda_A}
	\end{align}
	and $B_1$ is
	\begin{align}
		B_1 = B^{1/2} \label{eq:BEKK-lambda_B}.
	\end{align}
	The vec-reparameterization is therefore,
	\begin{align}
		\text{vec}(\Lambda_{0,t} ) = ( I_{p^2}- \tilde A_0 - \tilde B_0)\text{vec}(\Lambda_0)+\tilde A_0 \text{vec}(Y_{0,t-1}Y_{0,t-1}') +\tilde B_0 \text{vec}(\Lambda_{0,t-1}),
	\end{align}
	where
	\begin{align}
		\tilde A = \sum_{i=1}^{p^2} A_i^{\otimes2}, \ \ \ \  \tilde B = B_1^{\otimes2}.
		\label{eq:vec-lambda_AB}
	\end{align}
	We now use this reparameterzation of the model to rewrite $\text{vec}\left( \frac{1}{T}\sum_{t=1}^T\Lambda_{0,t}\right)$ as follows,
	\begin{align*}
		\text{vec}\left( \frac{1}{T}\sum_{t=1}^T\Lambda_{0,t}\right) &  =\left(I_{p^2}-\tilde A_0-\tilde B_0\right)\Lambda_0 +\tilde A_0\text{vec}\left(\frac{1}{T}\sum_{t=1}^TY_{0,t-1}Y_{0,t-1}'\right)+\tilde B_{0}\text{vec}\left(\frac{1}{T}\sum_{t=1}^T\Lambda_{0,t-1}\right) \\ 
		& = \left(I_{p^2}-\tilde A_0-\tilde B_0\right)\Lambda_0 +\tilde A_0\text{vec}\left(\frac{1}{T}\sum_{t=1}^TY_{0,t}Y_{0,t}'\right)+\tilde B_{0}\text{vec}\left(\frac{1}{T}\sum_{t=1}^T\Lambda_{0,t}\right) \\
		& \ \ \ \  +\frac{1}{T}\left(\tilde A_0(\text{vec}(Y_{0,0}Y_{0,0}')-\text{vec}(Y_{0,T}Y_{0,T}')) +\tilde B_0 (\text{vec}(\Lambda_{0,0})-\text{vec}(\Lambda_{0,T}))\right).
	\end{align*}
	Collecting terms, and noting that $\hat\Lambda=\frac{1}{T}\sum_{t=1}^TY_{0,t}Y_{0,t}'$, yields,
	\begin{align*}
		\text{vec}\left( \frac{1}{T}\sum_{t=1}^T\Lambda_{0,t}\right) & = [(I_{p^2}-\tilde B_0)]^{-1}(I_{p^2}-\tilde A_0-\tilde B_0)\text{vec}(\Lambda_0) + [(I_{p^2}-\tilde B_0)]^{-1} \tilde A_0 \text{vec}(\hat\Lambda)\\
		& \ \ \ \  +\frac{1}{T}[(I_{p^2}-\tilde B_0)]^{-1}\left(\tilde A_0(\text{vec}(Y_{0,0}Y_{0,0}')-\text{vec}(Y_{0,T}Y_{0,T}') +\tilde B_0 (\text{vec}(\Lambda_{0,0})-\text{vec}(\Lambda_{0,T}))\right),
	\end{align*}
	where $(I_{p^2}-\tilde B_0)$ is invertible since $B$ is diagonal with $\rho(B_0) = \rho(\tilde B_0)<1$. Insert this into \eqref{eq:mom_rewrite} and rearrange,
	\begin{align*}
		&\text{vec}(\hat\Lambda) - \text{vec}(\Lambda_0) =  
		\left(I_{p^2}-\tilde A_0-\tilde B_0 \right)^{-1}\left(I_{p^2}-\tilde B_0\right) \frac{1}{T}\sum_{t=1}^T(\Lambda_{0,t}^{1/2})^{\otimes2}\text{vec}(Z_tZ_t'-I_p) \\
		& 
		+\frac{1}{T}\left(I_{p^2}-\tilde A_0-\tilde B_0 \right)^{-1}\left(\tilde A_0(\text{vec}(Y_{0,0}Y_{0,0}')-\text{vec}(Y_{0,T}Y_{0,T}')) +\tilde B_0 (\text{vec}(\Lambda_{0,0})-\text{vec}(\Lambda_{0,T}))\right)
	\end{align*}
	By Markov's inequality it holds that for $\epsilon>0$,
	\begin{align*}
		& P\left(\Big|\Big|\frac{1}{T}(I_{p^2}-\tilde A_0-\tilde B_0)^{-1}\left(\tilde A_0(\text{vec}(Y_{0,0}Y_{0,0}')-\text{vec}(Y_{0,T}Y_{0,T}')) +\tilde B_0 (\text{vec}(\Lambda_{0,0})-\text{vec}(\Lambda_{0,T}))\right)\Big|\Big|>\epsilon\right) \\
		& \leq \frac{KE||X_t||^2}{T\epsilon}\rightarrow 0,
	\end{align*}
	as $T\rightarrow \infty$. This yields,
	\begin{align*}
		\text{vec}(\hat\Lambda)-\text{vec}(\Lambda_0) = &
		(I_{p^2}-\tilde A_0-\tilde B_0)^{-1}(I_{p^2}-\tilde B_0) \frac{1}{T}\sum_{t=1}^T(\Lambda_{0,t}^{1/2})^{\otimes2}\text{vec}(Z_tZ_t'-I_p) +o_p(T^{-1}).
	\end{align*}
	Recall that $\text{vec}(\hat H)-\text{vec}(H_0)=V_0^{\otimes2}\left(\text{vec}(\hat\Lambda)-\text{vec}(\Lambda_0)\right)$, and we find that,
	\begin{align}
		&\sqrt{T}(\text{vec}(\hat H)-\text{vec}(H_0)) = \notag \\
		& V_0^{\otimes2}(I_{p^2}-\tilde A_0-\tilde B_0)^{-1}(I_{p^2}-\tilde B_0) \frac{1}{\sqrt{T}}\sum_{t=1}^T(\Lambda_{0,t}^{1/2})^{\otimes2}\text{vec}(Z_tZ_t'-I_p) +o_p(T^{-1/2}).
		\label{eq:mom_done}
	\end{align}	
	As the model is parameterized in terms of the eigenvalues and -vectors, we now restate \eqref{eq:mom_done} in terms of $\lambda$ and $\upsilon$. Notice that $\hat\lambda-\lambda_0=D(V_0')^{\otimes2}\left(\text{vec}(\hat H)-\text{vec}(H_0)\right)$, where $D$ is a $p\times p^2$ matrix of zeros, apart from $p$ elements, $d_{i,i+(i-1)p}=1$ for $i=1, \hdots, p$, such that $D\text{vec}(\Lambda)=\lambda$, and we find that,
	\begin{align}
		\sqrt{T}(\hat\lambda-\lambda_0) = D(I_{p^2}-\tilde A_0-\tilde B_0)^{-1}(I_{p^2}-\tilde B_0) \frac{1}{\sqrt{T}}\sum_{t=1}^T(\Lambda_{0,t}^{1/2})^{\otimes2}\text{vec}(Z_tZ_t'-I_p) +o_p(T^{-1/2}).
		\label{eq:mom_lambda_done}
	\end{align}
	Next, since $\upsilon$ does not have a closed form solution as a function of $H$, we apply the mean-value theorem, and use the following result from \citet{magnus1985} (Theorem 1)
	\begin{align}
	\frac{\partial V_j}{\partial \text{vec}(H)} = V_{j}^\prime\otimes(\lambda_j I_p-H)^+
	\label{eq:V_der}
	\end{align}
	where $(\lambda_j I_p-H)^+$ is the Moore-Penrose (pseudo-) inverse of $(\lambda_j I_p-H)$. From this, we can apply the mean-value theorem to the $j$'th eigenvector, 
	\begin{align}
		\sqrt{T}(\hat V_j-V_{0,j}) & = V_{0,j}^\prime\otimes(\lambda_{0,j} I_p-H_0)^+(\text{vec}\sqrt{T}(\hat H)-\text{vec}(H_0)) \notag \\
		 & + \underbrace{(V_{0,j}^\prime\otimes(\lambda_{0,j} I_p-H_0)^+ - \tilde V_{j}^\prime\otimes(\tilde \lambda_{j} I_p-\tilde H_0)^+)\sqrt{T}(\text{vec}(\hat H)-\text{vec}(H_0))}_{=o_p(1)} \notag \\
		 & = V_{0,j}^\prime\otimes(\lambda_{0,j} I_p-H_0)^+\sqrt{T}(\text{vec}(\hat H)-\text{vec}(H_0)) + o_p(1), \ \ \ \ \ \ j =1,\hdots, p,
		\label{eq:mom_V_i_done}
	\end{align}
	where $\tilde H = \tilde V \tilde \Lambda \tilde V'$ is on the line between $H_0$ and $\hat H$. Hence, by \eqref{eq:mom_done} and \eqref{eq:mom_V_i_done},
	\begin{align}
		& \sqrt{T}(\hat\upsilon-\upsilon_0) = \notag \\
		&
		\begin{pmatrix}
			V_{0,1}^\prime\otimes(\lambda_{0,1} I_p-H_0)^+V_0^{\otimes2}(I_{p^2}-\tilde A_0-\tilde B_0)^{-1}(I_{p^2}-\tilde B_0) \frac{1}{\sqrt{T}}\sum_{t=1}^T(\Lambda_{0,t}^{1/2})^{\otimes2}\text{vec}(Z_tZ_t'-I_p) \\
			\vdots \\
			V_{0,p}^\prime\otimes(\lambda_{0,p} I_p-H_0)^+V_0^{\otimes2}(I_{p^2}-\tilde A_0-\tilde B_0)^{-1}(I_{p^2}-\tilde B_0) \frac{1}{\sqrt{T}}\sum_{t=1}^T(\Lambda_{0,t}^{1/2})^{\otimes2}\text{vec}(Z_tZ_t'-I_p) 
		\end{pmatrix}
		 +o_p(1)
		\label{eq:mom_V_done}
	\end{align}
	Finally, note that 
	\begin{align}
	\frac{\partial l^{(i)}_{t}(\theta_0^{(i)})}{\partial\kappa^{(i)}} & = -\frac{1}{\lambda_{i,t}(\theta_0^{(i)})}\frac{\partial\lambda_{i,t}(\theta_0^{(i)})}{\partial\kappa^{(i)}}(z_{i,t}^2-1).
	\label{eq:foc_kappa0}
	\end{align}
	Hence, by \eqref{eq:mom_lambda_done}, \eqref{eq:mom_V_done} and \eqref{eq:foc_kappa0}, we conclude that \eqref{eq:est_rewritten} holds.
\end{proof}

\begin{lemma}[Joint normality of parameter vector]
	Under Assumptions \ref{assumption1}-\ref{assumption5}
	\begin{align}
	 \sqrt{T}
	\begin{pmatrix}
		\hat\gamma-\gamma_0 \\
		\frac{\partial L_T^{(i)}(\theta_0)}{\partial\kappa}%\hat\kappa^{(i)}-\kappa^{(i)}_0
	\end{pmatrix} \overset{D}{\rightarrow} N(0,\Omega_0^{(i)}). \label{eq:norm_proof}
	\end{align} with $\Omega_0^{(i)}$ defined in \eqref{eq:asymp_var}.
	\label{lemma:norm1}
\end{lemma}
\begin{proof}
	Similar to the variance targeting literature (e.g. \citet{Pedersen2014} proof of Lemma B.8 and B.9), we use that \eqref{eq:est_rewritten} is a martingale difference sequence to show convergence in distribution. From \eqref{eq:est_rewritten}, define 
	\begin{align*}
		\underset{e\times e}{\mathcal{Y}_t}(\theta_0^{(i)})
		:=
		\begin{pmatrix}
		\mathcal{Y}_{1,t}(\theta_0^{(i)}) \\
		\mathcal{Y}_{2,t}(\theta_0^{(i)}) \\
		\mathcal{Y}_{3,t}(\theta_0^{(i)})
		%(V_0)^{\otimes2}\left[(I_{p^2}-\tilde A_0-\tilde B_0) \right]^{-1}(I_{p^2}-\tilde B_0)(\Lambda_{0,t}^{1/2})^{\otimes2}\text{vec}(Z_tZ_t'-I_p) \\
		%-\frac{1}{\lambda_{i,t}(\theta_0^{(i)})}\frac{\partial \lambda_{i,t}(\theta_0^{(i)})}{\partial \kappa^{(i)}}(z_{i,t}^2-1)
		\end{pmatrix},
	\end{align*}
	where each element is,
	\begin{align*}
		% \underset{p\times1}{\mathcal{Y}_{1,t}(\theta_0^{(i)})} & = D(I_{p^2}-\tilde A_0-\tilde B_0)^{-1}(I_{p^2}-\tilde B_0)(\Lambda_{0,t}^{1/2})^{\otimes2}\text{vec}(Z_tZ_t'-I_p)\\
		% \underset{p^2\times1}{\mathcal{Y}_{2,t}(\theta_0^{(i)})}& = 
		% \begin{pmatrix}
		% 	V_{0,1}\otimes(\lambda_{0,1}I_p-H_0)^+V_0^{\otimes2}(I_{p^2}-\tilde A_0-\tilde B_0)^{-1}(I_{p^2}-\tilde B_0)(\Lambda_{0,t}^{1/2})^{\otimes2}\text{vec}(Z_tZ_t'-I_p) \\
		% \vdots \\
		% V_{0,p}\otimes(\lambda_{0,p}I_p-H_0)^+V_0^{\otimes2}(I_{p^2}-\tilde A_0-\tilde B_0)^{-1}(I_{p^2}-\tilde B_0)(\Lambda_{0,t}^{1/2})^{\otimes2}\text{vec}(Z_tZ_t'-I_p)
		% \end{pmatrix}\\
		% \underset{(p+1)\times1}{\mathcal{Y}_{3,t}(\theta_0^{(i)})}& = -\frac{1}{\lambda_{i,t}(\theta_0^{(i)})}\frac{\partial \lambda_{i,t}(\theta_0^{(i)})}{\partial \kappa^{(i)}}(z_{i,t}^2-1).
		\underset{p\times1}{\mathcal{Y}_{1,t}(\theta_0^{(i)})} & = D(V_0')^{\otimes2} \Gamma_0(\Lambda_{0,t}^{1/2})^{\otimes2}\epsilon_t\\
		\underset{p^2\times1}{\mathcal{Y}_{2,t}(\theta_0^{(i)})}& = 
		\begin{pmatrix}
			% V_{0,1}\otimes(\lambda_{0,1}I_p-H_0)^+\Gamma_0(\Lambda_{0,t}^{1/2})^{\otimes2}\epsilon_t \\
			% \vdots \\
			% V_{0,p}\otimes(\lambda_{0,p}I_p-H_0)^+\Gamma_0(\Lambda_{0,t}^{1/2})^{\otimes2}\epsilon_t \\
			\chi_{0,1}\Gamma_0(\Lambda_{0,t}^{1/2})^{\otimes2}\epsilon_t \\
			\vdots \\
			\chi_{0,p}\Gamma_0(\Lambda_{0,t}^{1/2})^{\otimes2}\epsilon_t \\
		\end{pmatrix}\\
		\underset{(p+1)\times1}{\mathcal{Y}_{3,t}(\theta_0^{(i)})}& = -\frac{1}{\lambda_{i,t}(\theta_0^{(i)})}\frac{\partial \lambda_{i,t}(\theta_0^{(i)})}{\partial \kappa^{(i)}}(z_{i,t}^2-1),
	\end{align*}
	where we use the following definitions,
	\begin{align*}
		\epsilon_t & := \text{vec}(Z_tZ_t'-I_p), \\
		\Gamma_0 & := (V_0)^{\otimes2}\left[(I_{p^2}-\tilde A_0-\tilde B_0)\right]^{-1}(I_{p^2}-\tilde B_0)\\
		\chi_{0,i} & : = V_{0,i}'\otimes(\lambda_{0,i}I_p-H_0)^+
	\end{align*}
	Notice that
	\begin{align}
	E\left[\mathcal{Y}(\theta_0^{(i)})
	\right]
	= 0_{p(p+2)+1\times1}
	\end{align}
	since $Z_t$ is $iid$ with $E[Z_tZ_t']=I_p$. 

	Next, consider the covariance matrix,
	\begin{align}
		\underset{e\times e}{\Omega_0}^{(i)}:=
		E
		\left[
		\mathcal{Y}_t(\theta_0^{(i)}) \mathcal{Y}_t(\theta_0^{(i)})'		
		\right] 
		=
		E
		\left[
		\begin{pmatrix}
			\Omega_{0,11}^{(i)} & \Omega_{0,12}^{(i)} & \Omega_{0,13}^{(i)}\\
			(\Omega_{0,12}^{(i)})' & \Omega_{0,22}^{(i)} & \Omega_{0,23}^{(i)} \\
			(\Omega_{0,13}^{(i)})' & (\Omega_{0,23}^{(i)})' & \Omega_{0,33}^{(i)}\\
		\end{pmatrix}		
		\right] 
		\label{eq:asymp_var}
	\end{align}
	with 
	\begin{align*}
		&\Omega_{0,11}^{(i)} =
		D(V_0')^{\otimes2}
		\Gamma\left(\Lambda_{0,t}^{1/2}\right)^{\otimes2}\epsilon_t\epsilon_t'\left(\Lambda_{0,t}^{1/2}\right)^{\otimes2}\Gamma'(V_0)^{\otimes2}D',
		  \\
		%& [\Omega_{0,22}^{(i)}]_{mn} =
		%V_{0,m}\otimes(\lambda_{0,m}I_p-H_0)^+
		%\Gamma\left(\Lambda_{0,t}^{1/2}\right)^{\otimes2}\epsilon_t\epsilon_t'\left(\Lambda_{0,t}^{1/2}\right)^{\otimes2}\Gamma'(\lambda_{0,n}I_p-H_0)^+\otimes V_{0,n},
		&\Omega_{0,22}^{(i)} = %\\	&
		\begin{pmatrix}
			\chi_{0,1}\Gamma_0(\Lambda_{0,t}^{1/2})^{\otimes2}\epsilon_t\epsilon_t'(\Lambda_{0,t}^{1/2})^{\otimes2}\Gamma_0'\chi_{0,1}'& \hdots & \chi_{0,1}\Gamma_0(\Lambda_{0,t}^{1/2})^{\otimes2}\epsilon_t\epsilon_t'(\Lambda_{0,t}^{1/2})^{\otimes2}\Gamma_0'\chi_{0,p}' \\
			\vdots & \ddots & \vdots \\
			\chi_{0,p}\Gamma_0(\Lambda_{0,t}^{1/2})^{\otimes2}\epsilon_t\epsilon_t'(\Lambda_{0,t}^{1/2})^{\otimes2}\Gamma_0'\chi_{0,1}'& \hdots & \chi_{0,p}\Gamma_0(\Lambda_{0,t}^{1/2})^{\otimes2}\epsilon_t\epsilon_t'(\Lambda_{0,t}^{1/2})^{\otimes2}\Gamma_0'\chi_{0,p}' \\
		\end{pmatrix}
		  \\
		&\Omega_{0,33}^{(i)} = \frac{1}{\lambda_{i,t}^2(\theta_0^{(i)})}\frac{\partial \lambda_{i,t}(\theta_0^{(i)})}{\partial \kappa^{(i)}}\frac{\partial \lambda_{i,t}(\theta_0^{(i)})}{\partial \kappa^{(i)'}}(z_{i,t}^2-1)^2, \\
		& \Omega_{0,12}^{(i)} = D(V_0')^{\otimes2}
		\Gamma\left(\Lambda_{0,t}^{1/2}\right)^{\otimes2}\epsilon_t		
		\begin{pmatrix}	
			\epsilon_t'(\Lambda_{0,t}^{1/2})^{\otimes2}\Gamma_0'\chi_{0,1}' &
			\vdots &
			\epsilon_t'(\Lambda_{0,t}^{1/2})^{\otimes2}\Gamma_0'\chi_{0,p}' 
		\end{pmatrix} \\
		& \Omega_{0,13}^{(i)} = D(V_0')^{\otimes2}\Gamma\left(\Lambda_{0,t}^{1/2}\right)^{\otimes2}\epsilon_t\frac{1}{\lambda_{i,t}(\theta_0^{(i)})}\frac{\partial \lambda_{i,t}(\theta_0^{(i)})}{\partial \kappa^{(i)'}}(1-z_{i,t}^2), \\
		& \Omega_{0,23}^{(i)} = 		
		\begin{pmatrix}
			\chi_{0,1}\Gamma_0(\Lambda_{0,t}^{1/2})^{\otimes2}\epsilon_t \\
			\hdots \\
			\chi_{0,p}\Gamma_0(\Lambda_{0,t}^{1/2})^{\otimes2}\epsilon_t 
		\end{pmatrix}
		\frac{1}{\lambda_{i,t}(\theta_0^{(i)})}\frac{\partial \lambda_{i,t}(\theta_0^{(i)})}{\partial \kappa^{(i)'}}(1-z_{i,t}^2),
	\end{align*}
	which are $p\times p$, $p^2\times p^2$, $p+1\times p+1$, $p\times p^2$, $p\times p+1$ and $p^2\times p+1$ respectively. 

	To show that $\mathcal{Y}_t(\theta_0^{(i)})$ is square integrable, we verify that all elements of $\Omega_0^{(i)}$ are finite. By independence of $Z_t$ and $\Lambda_{0,t}$,
	\begin{align*}
		&E\left[|| \Omega_{0,11} || \right] \leq K E\left[\Big|\Big|\left(\Lambda_{0,t}^{1/2}\right)^{\otimes2}\Big|\Big|^2\right]E\left[||\epsilon_t||^2\right], \\
		&E\left[|| \Omega_{0,22} || \right] \leq K E\left[\Big|\Big|\left(\Lambda_{0,t}^{1/2}\right)^{\otimes2}\Big|\Big|^2\right]E\left[||\epsilon_t||^2\right],
	\end{align*}
	and using the euclidean matrix norm, $||A||=\sqrt{\text{tr}(A'A)}$ along with $\text{tr}(A\otimes B)=\text{tr}(A)\text{tr}(B)$ (for $A$ and $B$ square), 
	\begin{align*}
		E\left[\Big|\Big|\left(\Lambda_{0,t}^{1/2}\right)^{\otimes2}\Big|\Big|^2\right] = E\left[\text{tr}\left(\Lambda_{0,t}\right)^2\right] = E\left[\left(\sum_{i=1}^p\lambda_{i,t}(\theta_0^{(i)})\right)^2\right]\leq K<\infty,
	\end{align*}
	by Assumption \ref{assumption4}. Moreover, 
	\begin{align*}
		E\left[||\epsilon_t||^2\right] \leq E\left[||Z_t||^4\right]+K<\infty,
	\end{align*}
	as $E\left[||Z_t||^4\right]\leq KE\left[||X_t||^4\right]$. Hence $E\left[||\Omega^{(i)}_{0,11}||\right]<\infty$ and $E\left[||\Omega^{(i)}_{0,22}||\right]<\infty$. Next,
	\begin{align*}
		E\left[||\Omega_{0,33}^{(i)}||\right]\leq K<\infty,
	\end{align*}
	by Assumption \ref{assumption4}. Finally, $E\left[||\Omega_{0,12}^{(i)}||\right]$, $E\left[||\Omega_{0,13}^{(i)}||\right]$ and $E\left[||\Omega_{0,23}^{(i)}||\right]$ are finite by the Cauchy-Schwarz inequality.

	Because $\mathcal{Y}_t(\theta_0^{(i)})$ is a square integrable ergodic martingale difference sequence, we can invoke the central limit theorem for martingale differences, see e.g. \citet{Brown1971}, implying that \eqref{eq:norm_proof} holds.
\end{proof}

\begin{lemma}[Finite expectation of second order derivative]
	Under Assumptions \ref{assumption1}-\ref{assumption5}
	\begin{align*}
		E\left[\sup_{\theta^{(i)}\in\Theta^{(i)}}\left| \frac{\partial^2 l_t^{(i)}(\theta^{(i)})}{\partial\theta^{(i)}\partial\theta^{(i)'}}\right|\right]<\infty.
	\end{align*}
	\label{lemma:norm2}
\end{lemma}

\begin{proof}
	The second order derivative of $l_{i,t}(\theta^{(i)})$ is given in \eqref{eq:loglik_app_2der}, and we note that in order to show that $E\left[\sup_{\theta^{(i)}\in\Theta^{(i)}}\left| \frac{\partial^2 l_t^{(i)}(\theta^{(i)})}{\partial\theta^{(i)}\partial\theta^{(i)'}}\right|\right]<\infty$, if suffices to verify that,
	\begin{enumerate}
		\item 
		$E\left[\sup_{\theta^{(i)}\in\Theta^{(i)}}\left|
		\frac{y_{i,t}^2(\theta^{(i)})}{\lambda_{i,t}(\theta^{(i)})}\frac{1}{\lambda_{i,t}(\theta^{(i)})}\frac{\partial^2 \lambda_{i,t}(\theta^{(i)})}{\partial \theta^{(i)}\partial \theta^{(i)'}}
		\right|\right]<\infty$,
		
		\item 
		$E\left[\sup_{\theta^{(i)}\in\Theta^{(i)}}\left|
		\left( \frac{\partial^2 y_{i,t}(\theta^{(i)})}{\partial \theta^{(i)}_n \partial \theta^{(i)}_m}y_{i,t}(\theta^{(i)})	 + \frac{\partial y_{i,t}(\theta^{(i)})}{\partial \theta^{(i)}_n}\frac{\partial y_{i,t}(\theta^{(i)})}{ \partial \theta^{(i)}_m }
			\right)\frac{1}{\lambda_{i,t}(\theta^{(i)})}
		\right|\right]<\infty$,

		\item 
		$E\left[\sup_{\theta^{(i)}\in\Theta^{(i)}}\left|
		\left(\frac{\partial \lambda_{i,t}(\theta^{(i)})}{\partial\theta^{(i)}_n}\frac{\partial y_{i,t}(\theta^{(i)})}{\partial\theta^{(i)}_m}+\frac{\partial \lambda_{i,t}(\theta^{(i)})}{\partial\theta^{(i)}_m}\frac{\partial y_{i,t}(\theta^{(i)})}{\partial\theta^{(i)}_n}
		\right)\frac{y_{i,t}(\theta^{(i)})}{\lambda_{i,t}^2(\theta^{(i)})}
		\right|\right]<\infty$,

		\item 
		$E\left[\sup_{\theta^{(i)}\in\Theta^{(i)}}\left|
		\frac{y_{i,t}^2(\theta^{(i)})}{\lambda_{i,t}(\theta^{(i)})}\frac{1}{\lambda_{i,t}^2(\theta^{(i)})}\frac{\partial \lambda_{i,t}(\theta^{(i)})}{\partial \theta^{(i)}_n}\frac{\partial \lambda_{i,t}(\theta^{(i)})}{\partial \theta^{(i)}_m} 
		\right|\right]<\infty$.
	\end{enumerate}
	Inequalities 1.-3. are finite by Assumption \ref{assumption4}. The last inequality requires finite $2+s$ moments for $s>0$, as,
	\begin{align*}
		E\left[\sup_{\theta^{(i)}\in\Theta^{(i)}}\left| 
		\frac{\partial \lambda_{i,t}(\theta^{(i)})}{\partial \theta^{(i)}_n}
		\frac{1}{\lambda_{i,t}(\theta^{(i)})}
		\right|\right]\leq K.
	\end{align*}
	To see this, recall that the process for the eigenvalues can be written as,
	\begin{align*}
		\lambda_{i,t} = \sum_{l=0}^\infty b_i^l\big(w_i+\sum_{j=1}^pa_{ij}y_{j,t-l-1}^2(\theta^{(i)})\big) =\frac{w_i}{1-b_i}+\sum_{l=0}^\infty\sum_{j=1}^pb_i^la_{ij}y_{j,t-l-1}^2(\theta^{(i)})
	\end{align*}
	with $w_i = (1-b_i)\lambda_i-\sum_{j=1}^pa_{ij}\lambda_j$ and $\theta^{(i)}=[ \gamma',a_{i1}, \hdots, a_{ip}, b_i ]'$. The derivatives of the eigenvalues are,
	\begin{align}
		\frac{\partial \lambda_{i,t}(\theta^{(i)})}{\partial \gamma_k} & = \frac{\partial w_i(1-b_i)^{-1}}{\partial \gamma_k}+2\sum_{l=0}^\infty\sum_{j=1}^pb_i^la_{ij}y_{j,t-l-1}(\theta^{(i)})\frac{\partial V_i}{\partial \gamma_k}'X_t,  \label{eq:lambda_der_g}
		\\
		\frac{\partial \lambda_{i,t}(\theta^{(i)})}{\partial a_{ij}} & = %\frac{-\lambda_j}{1-b_i}
		\frac{\partial w_i(1-b_i)^{-1}}{\partial a_{ij}}+\sum_{l=0}^\infty b_i^ly_{j,t-l-1}^2(\theta^{(i)}), \label{eq:lambda_der_a}
		\\
		\frac{\partial \lambda_{i,t}(\theta^{(i)})}{\partial b_{i}} & = %\frac{-\sum_{j=1}^pa_{ij}\lambda_j}{(1-b_i)^2}
		\frac{\partial w_i(1-b_i)^{-1}}{\partial b_{i}}+\sum_{l=1}^\infty\sum_{j=1}^p lb_i^{l-1}a_{ij}y_{j,t-l-1}^2(\theta^{(i)}). \label{eq:lambda_der_b}
	\end{align}
	By Assumption \ref{assumption5} along with the inequality $z/(1+z)\leq z^s$ for $s\in(0,1)$ for all $z\geq0$, for all interior $\theta^{(i)}\in\Theta^{(i)}$,
	\begin{align}
	\sup_{\theta^{(i)}\in\Theta^{(i)}}	\left|\frac{1}{\lambda_{i,t}(\theta^{(i)})}\frac{\partial \lambda_{i,t}(\theta^{(i)})}{\partial \gamma_j} \right| & \leq 
	\sup_{\theta^{(i)}\in\Theta^{(i)}} \left|K +K\sum_{l=0}^\infty\frac{\sum_{j=1}^pb_i^{l}a_{ij}y_{j,t-l-1}^{2}(\theta^{(i)})}{K+\sum_{j=1}^pb_i^{l}a_{ij}y_{j,t-l-1}^{2}(\theta^{(i)})} \right| 
	\notag \\ & \leq 
	\sup_{\theta^{(i)}\in\Theta^{(i)}}\left| K +K\sum_{l=0}^\infty\sum_{j=1}^pb_i^{ls}a_{ij}^sy_{j,t-l-1}^{2s}(\theta^{(i)})\ \right| \leq K, \label{eq:ratio_bounded1}
	\end{align}
	using $\sup_{\theta^{(i)}\in\Theta^{(i)}}|\frac{\partial y_{i,t-1}^2(\theta^{(i)})}{\partial \gamma_j}|=\sup_{\theta^{(i)}\in\Theta^{(i)}}|2V_i(\theta^{(i)})'X_{t-1}\frac{\partial V_i'(\theta^{(i)})}{\partial \gamma_j}X_{t-1}| \leq \sup_{\theta^{(i)}\in\Theta^{(i)}}|K y_{i,t-1}^2|$, due to the orthonormality of $V_i$. Similarly,
	\begin{align}
	\sup_{\theta^{(i)}\in\Theta^{(i)}}	\left|\frac{1}{\lambda_{i,t}(\theta^{(i)})}\frac{\partial \lambda_{i,t}(\theta^{(i)})}{\partial a_{ij}} \right| & \leq 
	\sup_{\theta^{(i)}\in\Theta^{(i)}}\left|K+K\frac{1}{a_{ij}}\sum_{l=0}^\infty\frac{b_i^l a_{ij}y^2_{j,t-l-1}(\theta^{(i)})}{K+b_i^la_{ij}y_{j,t-l-1}^2(\theta^{(i)})}\right | \notag \\ & \leq 
	\sup_{\theta^{(i)}\in\Theta^{(i)}} \left|K+K\sum_{l=0}^\infty b_i^{ls} a_{ij}^s y_{j,t-l-1}^{2s}(\theta^{(i)})\right|\leq K, \label{eq:ratio_bounded2}	\\
	\sup_{\theta^{(i)}\in\Theta^{(i)}}	\left|\frac{1}{\lambda_{i,t}(\theta^{(i)})}\frac{\partial \lambda_{i,t}(\theta^{(i)})}{\partial b_i} \right| & \leq 
	\sup_{\theta^{(i)}\in\Theta^{(i)}}\left|K+K\frac{1}{b_i}\sum_{l=1}^\infty l\frac{\sum_{j=1}^pb_i^{l}a_{ij}y^2_{j,t-l-1}(\theta^{(i)})}{K+\sum_{j=1}^pb_i^la_{ij}y_{j,t-l-1}^2(\theta^{(i)})}\right| \notag \\ &\leq 
	\sup_{\theta^{(i)}\in\Theta^{(i)}}\left|K+K\sum_{l=1}^\infty\sum_{j=1}^plb_i^{sl}a_{ij}^sy_{j,t-l-1}^{2s}(\theta^{(i)})\right|\leq K, \label{eq:ratio_bounded3}
	\end{align}
	such that \eqref{eq:ratio_bounded1}-\eqref{eq:ratio_bounded3} are finite when $\{X_t\}_{t\in\mathbb{N}}$ is stationary, ergodic and has finite fractional moments.
 	Hence the last inequality is shown to hold under Assumptions \ref{assumption1}-\ref{assumption5}.	
\end{proof}

\begin{lemma}[Uniform convergence of second order derivative]
	Under Assumptions \ref{assumption1}-\ref{assumption5}, for $T\rightarrow\infty$,
	\begin{align*}
		\sup_{\theta^{(i)}\in\Theta^{(i)}}\left| 
		\frac{\partial^2 L_t^{(i)}(\theta^{(i)})}{\partial\theta^{(i)}\partial\theta^{(i)'}}		-	E\left[\frac{\partial^2 l_t^{(i)}(\theta^{(i)})}{\partial\theta^{(i)}\partial\theta^{(i)'}}\right]
		\right|\overset{a.s.}{\rightarrow}0.
	\end{align*}
	\label{lemma:norm3}
\end{lemma}
\begin{proof}
	Since $\frac{\partial^2l_{i,t}(\theta^{(i)})}{\partial\theta^{(i)}\partial\theta^{(i)'}}$ is a function of $(X_t, X_{t-1},\hdots)$ and $\theta^{(i)}$, it is strictly stationary and ergodic. The result then follows by Lemma \ref{lemma:norm2} and the uniform law of large numbers for stationary and ergodic processes, see Theorem A.2.2 of \citet{White1994}.
\end{proof}

\begin{lemma}[Non-singular $J_0^{(i)}$]
	Under Assumption \ref{assumption1}-\ref{assumption5}, $J_0^{(i)}$, given in \eqref{eq:appendix_asympvar}, is non-singular.\label{lemma:norm4}
\end{lemma}
\begin{proof}
	$J_0^{(i)}$ is identical to the Hessian of a univariate (extended) GARCH model (for the rotated returns). The non-singularity of $J_0^{(i)}$ therefore follows from \citet{berkes2003} Lemma 5.7. 

	% Suppose that $J_0^{(i)}$ is indeed singular, then there exists a $(p+1)\times 1$ vector $c\neq0$ s.t. $c' J_0^{(i)}=0$. Based on \eqref{eq:nonsing1}, this implies that 
	% \begin{align}
	% 	c'\frac{\partial\lambda_{i,t}(\theta_0^{(i)})}{\partial \kappa^{(i)}_n}=0, \ \ \ \ \ \ \ \  \text{ a.s.}
	% 	\label{eq:nonsing2}
	% \end{align}
	% as $\lambda_{i,t}(\theta_0^{(i)})>0$. By the strict stationarity of $\{X_t\}_{t\in\mathbb{N}}$, \eqref{eq:nonsing2} is,
	% \begin{align}
	% 	\tilde w_i +\sum_{j=1}^p\tilde a_{ij}y_{j,t-1}^2+\tilde b_i\lambda_{i,t-1} = 0, \ \ \ \ \ \ \ \ \text{ a.s.} \label{eq:nonsing3}
	% \end{align}
	% where
	% \begin{align}
	% 	\tilde w_i = \sum_{n=1}^{p+1}c_n\frac{\partial w_i}{\partial \kappa_n}, \ \ \ \ \tilde a_i = \sum_{n=1}^{p+1}c_n\frac{\partial a_{ij}}{\partial \kappa_n}, \ \ \ \  \tilde b_i = \sum_{n=1}^{p+1}c_n\frac{\partial b_i}{\partial \kappa_n}.
	% \end{align}
	% Subtract \eqref{eq:nonsing3} from $\lambda_{i,t}(\theta_0^{(i)})$,
	% \begin{align}
	% 	\lambda_{i,t}(\theta_0^{(i)}) = (w_i-\tilde w_i) + \sum_{j=1}^p(a_{ij}-\tilde a_{ij})y_{j,t-1}^2+(b_i-\tilde b_i)\lambda_{i,t-1}.
	% \end{align}
	% Since $c\neq 0$, we have found an equivalent representation of $\lambda_{i,t}(\theta_0^{(i)})$. However, this is in contradiction with assumption \ref{assumption3}. Therefore, $J_0^{(i)}$ must be non-singular.
\end{proof}

\begin{lemma}[Asymptotic irrelevance of (fixed) initial values]
	Under Assumptions \ref{assumption1}-\ref{assumption5}, %and for $T\rightarrow\infty$,
	\begin{align}
		&\left|\frac{1}{\sqrt{T}}\sum_{t=1}^T\left(\frac{\partial l_t^{(i)}(\theta_0^{(i)})}{\partial \kappa_n^{(i)}}- \frac{\partial l_{t,h}^{(i)}(\theta_0^{(i)})}{\partial \kappa_n^{(i)}}\right)\right|\overset{p}{\rightarrow}0 \label{eq:initial_val_1}
	\end{align}
	for $n = 1,\hdots, p+1$, and
	\begin{align}
		\sup_{\theta^{(i)}\in\Theta^{(i)}}&\left|\frac{1}{{T}}\sum_{t=1}^T\left(\frac{\partial^2 l_t^{(i)}(\theta^{(i)})}{\partial \theta_n^{(i)}\partial \theta_m^{(i)}}- \frac{\partial^2 l_{t,h}^{(i)}(\theta^{(i)})}{\partial \theta_n^{(i)}\partial \theta_m^{(i)}}\right)\right|\overset{p}{\rightarrow}0, \label{eq:initial_val_2}
	\end{align}
	for $n,m=1,\hdots, p(p+1)/2+p+1$.
	\label{lemma:norm5}
\end{lemma}
\begin{proof}
	Consider first \eqref{eq:initial_val_1} concerning the elements of the score vector. By \eqref{eq:loglik_app_1der}, the triangle inequality and $\sup_{\theta^{(i)}\in \Theta^{(i)}}|y_{i,t}|=\sup_{\gamma\in \mathcal{H}}|V_i'(\gamma)X_t|\leq K||X_t||$, 
	\begin{align*}
		&\left|\frac{1}{\sqrt{T}}\sum_{t=1}^T\left(\frac{\partial l_t^{(i)}(\theta_0^{(i)})}{\partial \kappa_n^{(i)}}- \frac{\partial l_{t,h}^{(i)}(\theta_0^{(i)})}{\partial \kappa_n^{(i)}}\right)\right| 
	%	=\Bigg|\frac{1}{\sqrt{T}}\sum_{t=1}^T\Bigg(\frac{1}{\lambda_{i,t}(\theta_0^{(i)})}\frac{\partial \lambda_{i,t}(\theta_0^{(i)})}{\partial \kappa^{(i)}_n} -
	%\frac{1}{\lambda_{i,t,h}(\theta_0^{(i)})}\frac{\partial \lambda_{i,t,h}(\theta_0^{(i)})}{\partial \kappa^{(i)}_n}
	% \\ &  +
	%y_{i,t}^2(\theta_0^{(i)})\Bigg(\frac{1}{\lambda^2_{i,t}(\theta_0^{(i)})}\frac{\partial \lambda_{i,t}(\theta_0^{(i)})}{\partial \kappa^{(i)}_n} -\frac{1}{\lambda^2_{i,t,h}(\theta_0^{(i)})}\frac{\partial \lambda_{i,t,h}(\theta_0^{(i)})}{\partial \kappa^{(i)}_n}\Bigg)\Bigg)\Bigg| 
	\leq 
	 \\
	 &
	 \frac{1}{\sqrt{T}}\sum_{t=1}^T\Bigg|\frac{1}{\lambda_{i,t}(\theta_0^{(i)})}\frac{\partial\lambda_{i,t}(\theta_0^{(i)})}{\partial\kappa^{(i)}_n}-\frac{1}{\lambda_{i,t,h}(\theta_0^{(i)})}\frac{\partial\lambda_{i,t,h}(\theta_0^{(i)})}{\partial\kappa^{(i)}_n}\Bigg| +
	 \\
	 & K\frac{1}{\sqrt{T}}\sum_{t=1}^T||X_t||^2\Bigg|\frac{1}{\lambda_{i,t}(\theta_0^{(i)})}\frac{\partial\lambda_{i,t}(\theta_0^{(i)})}{\partial\kappa^{(i)}_n}-\frac{1}{\lambda_{i,t,h}(\theta_0^{(i)})}\frac{\partial\lambda_{i,t,h}(\theta_0^{(i)})}{\partial\kappa^{(i)}_n}
	 \Bigg|
	\end{align*}
	Following \citet{francq2011} (p. 649), 
	\begin{align}
		&\sup_{\theta^{(i)}\in\Theta^{(i)}}\left|\frac{\partial\lambda_{i,t}(\theta^{(i)})}{\partial\theta^{(i)}_n}-\frac{\partial\lambda_{i,t,h}(\theta^{(i)})}{\partial\theta^{(i)}_n}\right|= \notag \\ &\sup_{\theta^{(i)}\in\Theta^{(i)}}\left|\frac{\partial b_i^{t}}{\partial \theta^{(i)}_n}\left(\lambda_{i,0}(\theta^{(i)})-\lambda_{i,0,h}(\theta^{(i)})\right)+b_i^t\frac{\partial\left(\lambda_{i,0}(\theta^{(i)})-\lambda_{i,0,h}(\theta^{(i)})\right)}{\partial\theta^{(i)}_n}\right|\leq %\notag \\ & 
		K\phi^t. %+o(1), 
		\label{eq:lemma9_1}
	\end{align}
	as both $b_i^t\overset{ }{\rightarrow} 0$ and $\partial b_i^t/\partial b_i\overset{ }{\rightarrow} 0$ for $t\rightarrow\infty$. Furthermore, notice that
	\begin{align}
	 	\sup_{\theta^{(i)}\in\Theta^{(i)}}\left|\frac{1}{\lambda_{i,t}(\theta^{(i)})}-\frac{1}{\lambda_{i,t,h}(\theta^{(i)})}\right|\leq\sup_{\theta^{(i)}\in\Theta^{(i)}}\left|\frac{1}{w_i}-\frac{1}{w_i}\right|\leq K.\label{eq:lemma9_2}
	\end{align} 
	Using \eqref{eq:lemma9_1}-\eqref{eq:lemma9_2}, we find that \ref{eq:initial_val_1} can be bounded as,
	\begin{align*}
	& \Bigg|\frac{1}{\sqrt{T}}\sum_{t=1}^T\Bigg(\frac{\partial l_t^{(i)}(\theta_0^{(i)})}{\partial \kappa_n^{(i)}}- \frac{\partial l_{t,h}^{(i)}(\theta_0^{(i)})}{\partial \kappa_n^{(i)}}\Bigg)\Bigg| 
	 \leq K\frac{1}{\sqrt{T}}\sum_{t=1}^T\phi^t\left(1+||X_t||^2 \right). %+o(1).
	\end{align*}
	As $\sum_{t=1}^T\phi^t\rightarrow(1-\phi)^{-1}$ it holds that $KT^{-1/2}\sum_{t=1}^T\phi_t\rightarrow 0$ for $T\rightarrow \infty$, and by the Markov inequality for $\epsilon>0$,
	\begin{align*}
		P\left(\frac{1}{\sqrt{T}}\sum_{t=1}^T\phi^t(1+||X_t||^2)>\epsilon\right)\leq \epsilon^{-1}(1+E\left|\left|X_t\right|\right|^2)\frac{1}{\sqrt{T}}\sum_{t=1}^T\phi^t\overset{p}{\rightarrow} 0,
	\end{align*}
	as $E||X_t||^2<\infty$ by Assumption \ref{assumption1}. This proves the first statement of the lemma.

	Next, we consider the requirement in \eqref{eq:initial_val_2}. 	The second order derivative of the log-likelihood function is given in \eqref{eq:loglik_app_2der}, and to show that the expression in \eqref{eq:initial_val_2} converges to zero in probability we need three additional results. First,
	\begin{align}
		&\sup_{\theta^{(i)}\in\Theta^{(i)}}\left|\frac{\partial^2\lambda_{i,t}(\theta^{(i)})}{\partial\theta^{(i)}_n\partial\theta^{(i)}_m}-\frac{\partial^2\lambda_{i,t,h}(\theta^{(i)})}{\partial\theta^{(i)}_n\partial\theta^{(i)}_m}\right|= \notag 
		\\ 
		&\sup_{\theta^{(i)}\in\Theta^{(i)}}\Bigg|\frac{\partial^2 b_i^{t}}{\partial \theta^{(i)}_n\partial \theta^{(i)}_m}\left(\lambda_{i,0}(\theta^{(i)})-\lambda_{i,0,h}(\theta^{(i)})\right)+b_i^t\frac{\partial^2\left(\lambda_{i,0}(\theta^{(i)})-\lambda_{i,0,h}(\theta^{(i)})\right)}{\partial\theta^{(i)}_n\partial\theta^{(i)}_m}
		\notag \\
		&
		\ \ \ \ \ \ \ \ \ \ \ \  +\frac{\partial b_i^t}{\partial\theta_n^{(i)}}\frac{\partial\left(\lambda_{i,0}(\theta^{(i)})-\lambda_{i,0,h}(\theta^{(i)})\right)}{\partial\theta^{(i)}_m}
		+\frac{\partial b_i^t}{\partial\theta_m^{(i)}}\frac{\partial\left(\lambda_{i,0}(\theta^{(i)})-\lambda_{i,0,h}(\theta^{(i)})\right)}{\partial\theta^{(i)}_n}
		\Bigg|
		\leq %\notag 
		%\\
		%& \ \ \ \ \ \ \ \ \ \ \ \ \ 
		K\phi^t. %+o(1).
		\label{eq:lemma9_3}
	\end{align}
	Second, by \eqref{eq:ratio_bounded1}-\eqref{eq:ratio_bounded3},
	\begin{align}
		\sup_{\theta^{(i)}\in\Theta^{(i)}}\left|\frac{1}{\lambda_{i,t}(\theta^{(i)})}\frac{\partial\lambda_{i,t}(\theta^{(i)})}{\partial\theta^{(i)}_n}\right|\leq K, 
		& \ \ \ \  
		\sup_{\theta^{(i)}\in\Theta^{(i)}}\left|\frac{1}{\lambda_{i,t,h}(\theta^{(i)})}\frac{\partial\lambda_{i,t,h}(\theta^{(i)})}{\partial\theta^{(i)}_n}\right|\leq K,
		\label{eq:lemma9_4new}
	\end{align}
	Third,
	\begin{align}
		\sup_{\theta^{(i)}\in\Theta^{(i)}}\left|\frac{1}{\lambda_{i,t}(\theta^{(i)})}-\frac{1}{\lambda_{i,t,h}(\theta^{(i)})}\right|=
		\sup_{\theta^{(i)}\in\Theta^{(i)}}\left|\frac{1}{\lambda_{i,t}(\theta^{(i)})}(\lambda_{i,t}(\theta^{(i)})-\lambda_{i,t,h}(\theta^{(i)}))\frac{1}{\lambda_{i,t,h}(\theta^{(i)})}\right|\leq K\phi^t.
		\label{eq:lemma9_5new}
	\end{align}
	Hence, by \eqref{eq:lemma9_1}-\eqref{eq:lemma9_5new} along with the triangle inequality, \eqref{eq:initial_val_2} can be bounded as,
	\begin{align*}
		&\sup_{\theta^{(i)}\in\Theta^{(i)}}\left|\frac{1}{{T}}\sum_{t=1}^T\left(\frac{\partial^2 l_t^{(i)}(\theta^{(i)})}{\partial \theta_n^{(i)}\partial \theta_m^{(i)}}- \frac{\partial^2 l_{t,h}^{(i)}(\theta^{(i)})}{\partial \theta_n^{(i)}\partial \theta_m^{(i)}}\right)\right| \leq
		K\frac{1}{T}\sum_{t=1}^T\phi^t\left(1+||X_t||^2+||X_t||^{4}\right),
	\end{align*}
	which by Markov's inequality converges to zero in probability,
	\begin{align*}
		P\left(K\frac{1}{T}\sum_{t=1}^T\phi^t\left(1+||X_t||^2+||X_t||^4\right)>\epsilon\right)\leq\epsilon^{-1}\left(K(1+E[||X_t||^2]+E[||X_t||^4])\frac{1}{T}\sum_{t=1}^T\phi^t\right)\overset{p}{\rightarrow} 0,
	\end{align*}
	which concludes the proof, and we conclude that the (fixed) initial values used in the estimation do not matter for $T\rightarrow \infty$.

\end{proof}

\section{Stationarity, ergodicity and existence of moments}
\label{appendix:stat_erg_mom}
The following two lemmata provide sufficient and necessary conditions for strict stationarity, ergodicity, and finite moments for the multivariate GARCH model and are both stated without proofs. They were originally stated in the context of the extended CCC model but are also applicable in the present model, as the $\lambda$-GARCH, conditional on the eigenvectors, is an extended CCC model for the rotated returns.

Rewrite the process of the eigenvalues as a stochastic recurrence equation,
\begin{align*}
 	\lambda_t = W+\mathcal{A}_{t-1}\lambda_{t-1},
\end{align*} 
where $\mathcal{A}_{t-1}=A\text{ diag}\left(Z_{t-1}^{\odot2}\right)+B$ is an $iid$ $p\times p$ sequence for $t\in\mathbb{Z}$.

\begin{lemma}[\citet{francq2019} Theorem 10.6]
\label{lemma:stationarity}
	A necessary and sufficient condition for the existence of a unique, non-anticipative, strictly stationary and ergodic solution to the process $(X_t: t\in\mathbb{Z})$ is $\gamma<0$,  with $\gamma$ defined as the top Lyapunov coefficient, $\gamma = \underset{ t \rightarrow \infty}{\lim}\frac{1}{t}E\left[\log||\prod_{i=1}^t \mathcal{A}_i||\right]$
\end{lemma}
Notice that Lemma \ref{lemma:stationarity} only ensures the existence of fractional moments, $E||X_t||^s<\infty$, $0<s<1$. We next restate a result from %\citet{Ling2003} (theorem 2.2) and 
\citet{Pedersen2017b} (Proposition 2.1), which contain necessary and sufficient conditions for finite (non-fractional) moments.
\begin{lemma}[\citet{Pedersen2017b} Proposition 2.1]
\label{lemma:moments}
	Let $(X_t: t\in\mathbb{Z})$ denote a strictly stationary and ergodic process. Then $E\left[ ||X_t^{\odot2}||^k \right]<\infty$, $k\in\mathbb{Z}$ if and only if $\rho\left(E\left[\mathcal{A}_{t-1}^{\otimes k}\right]\right)<1$.
\end{lemma}

%%%%%%%%%%%%%%%%%%%%%%%%%%%%%%%%%%%%% %%%%%%%%%%%%%%%%%%%%%%%%%%%%%%%%%%%%% %%%%%%%%%%%%%%%%%%%%%%%%%%%%%%%%%%%%% 
%%%%%%%%%%%%%%%%%%%%%%%%%%%%%%%%%%%%%       OLD SIMULATION SECTION BEGIN    %%%%%%%%%%%%%%%%%%%%%%%%%%%%%%%%%%%%%
%%%%%%%%%%%%%%%%%%%%%%%%%%%%%%%%%%%%% %%%%%%%%%%%%%%%%%%%%%%%%%%%%%%%%%%%%% %%%%%%%%%%%%%%%%%%%%%%%%%%%%%%%%%%%%% 
% \begin{itemize}
% 	\item Drop section on simulation as it is now, replace with comparisons of QMLE and SVTE for $T=1000$, $T=5000$ for $\eta_t$ Gaussian, $v=9$ for $p=2$ (and $p=5$?) for two diagonal models, with finite fourth order moments.
% \end{itemize}

\section{Simulation study}
\label{appendix:simulation}
This appendix illustrates the theoretical results through simulations: we simulate the large-sample distribution of the STE in three cases:  In the first, we illustrate the sufficiency of finite fourth order moments, and show that both steps of the STE are consistent and asymptotically normal when $E||X_t||^4<\infty$. The second case considers the distribution of the STE when the data-generating process (DGP) does not admit finite fourth order moments, but rather has finite second order moments, $E||X_t||^2<\infty$, indicating that the STE should be consistent, but have a non-normal limiting distribution. Finally, the third simulation considers the STE when the DGP only admits a finite mean, $E||X_t||<\infty$.

\subsection{Case 1: The DGP satisfies the sufficient condition for asymptotic normality}
\begin{figure}[ht]
	\centering
	\includegraphics[width=\textwidth]{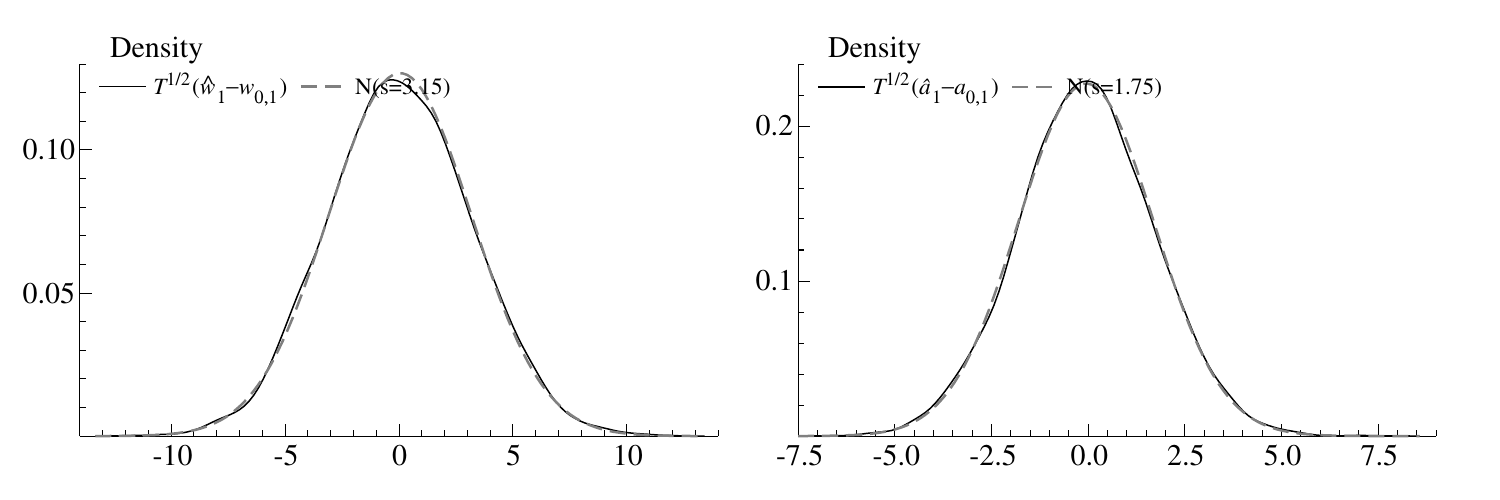}
	%\caption{Densities of estimated parameters when $E||X_t||^4<\infty$. In the density plots the solid line is the estimated density, and the grey dashed line is the normal distribution. The QQ-plots compare the quantiles of the estimated parameters with the ones of a normal distribution (grey line). The dashed grey lines are the asymptotic 95\% standard error bands of a normal distribution.
	\caption{Densities of estimated parameters when $E||X_t||^4<\infty$. The solid line is the estimated density, and the grey dashed line is the normal distribution.	}
	\label{fig:sim1}
\end{figure}
Consider the bivariate $\lambda-$GARCH with Gaussian innovations,
\begin{align}
	X_t = V\Lambda_t^{1/2}Z_t, \ \ \eta_t \  iid \ N(0,I_2), \ \ \ \Lambda_t=\text{diag}(\lambda_t), \ \ \ \lambda_t=W+AY_{t-1}^{\odot2}+B\lambda_{t-1},
	\label{eq:sim1_1}
\end{align}
with parameters 
\begin{align}
	V_0 = 
	\begin{pmatrix}
		0.89 & 0.45 \\
       -0.45 & 0.89
	\end{pmatrix}, 	 \ \ \
	W_0 = \begin{pmatrix}
	1.5 \\
	0.46
	\end{pmatrix}, \ \ \ 
	A_0 =
	\begin{pmatrix}
		0.33 & 0 \\
		0 & 0.25
	\end{pmatrix}, \ \ \
	B_0=0_{2\times2},
	\label{eq:sim1_2}
\end{align}
such that $\rho(E[A\text{ diag}(Z_t^2)+B]^{\otimes2} )=\max \left(a^2_{ii}E[z_{i,t}^4]\right)<1$ for $i=1,2$. For $Z_t \ iid \ N(0,I_2)$ this corresponds to $\max(a_i)<1/\sqrt{3}$, and by Lemma \ref{lemma:moments} (with $k=2$), the stationary solution of the process has finite fourth order moments, and the moment restrictions of Theorem \ref{theorem:asympnorm} are satisfied. 

We simulate $N=10.000$ realizations of \eqref{eq:sim1_1}-\eqref{eq:sim1_2} with $T=10.000$ observations, and estimate $W$ and $A$ using STE. Figure \ref{fig:sim1} contains standardized densities %and QQ-plots 
of $w_1$ and $a_{11}$.  The figure suggests that the STE is indeed consistent and asymptotically normal, in line with the findings in Theorem \ref{theorem:asympnorm}. 

\subsection{Case 2: The DGP satisfies the sufficient condition for consistency}
Next, we consider the case where the DGP has finite second order moments, but does not admit finite fourth order moments. We consider \eqref{eq:sim1_1}, with parameters
%Consistent, but not (exactly) Gaussian, almost opposite to \citet{Pedersen2014}
\begin{align}
	V_0 = 
	\begin{pmatrix}
		0.89 & 0.45 \\
       -0.45 & 0.89
	\end{pmatrix}, 	 \ \ \
	W_0 = \begin{pmatrix}
	1.5 \\
	0.46
	\end{pmatrix}, \ \ \ 
	A_0 =
	\begin{pmatrix}
		0.60 & 0 \\
		0 & 0.55
	\end{pmatrix}, \ \ \
	B_0=0_{2\times2},
	\label{eq:sim2_2}
\end{align}
such that $\rho(E[A_0\text{ diag}(Z_t^2)+B_0])=\max \left(a^2_{0,ii}\right)<1$ for $i=1,2$. By Lemma \ref{lemma:moments} (with $k=1$), the stationary solution of the process admits finite second order moments, and the moment restrictions for asymptotic normality (Theorem \ref{theorem:asympnorm}) are not satisfied. However, by Theorem \ref{theorem:consistency}, the estimator should be consistent.

\begin{figure}[ht]
	\centering
	\includegraphics[width=\textwidth]{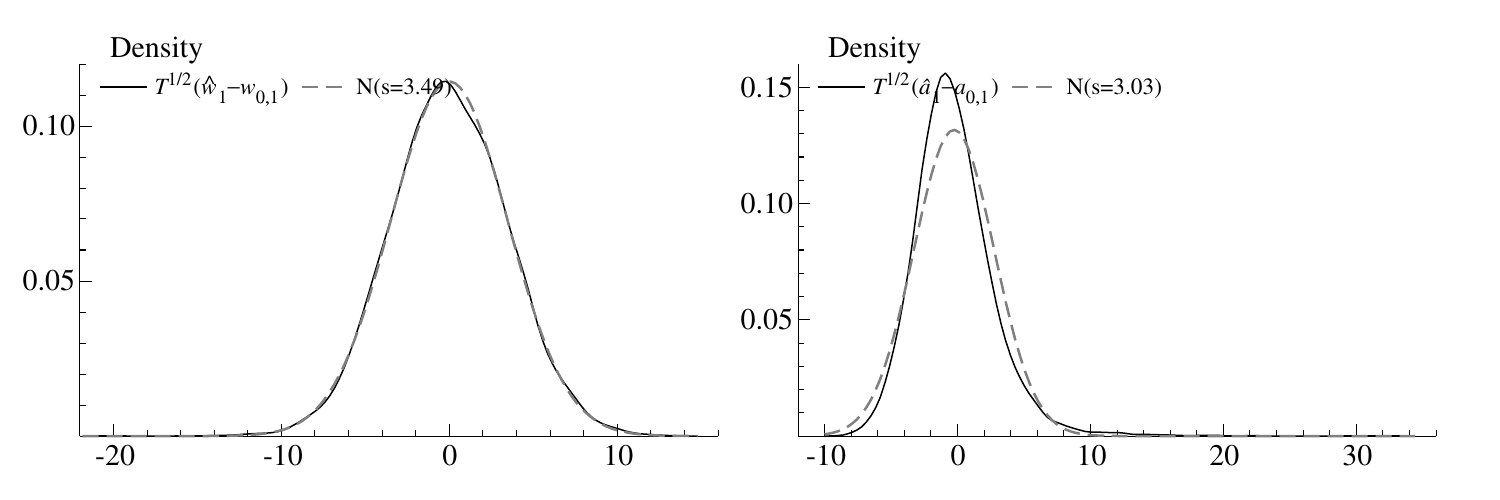}
	%\caption{Densities of estimated parameters when $E||X_t||^2<\infty$.  In the density plots the solid line is the estimated density, and the grey dashed line is the normal distribution. The QQ-plots compare the quantiles of the estimated parameters with the ones of a normal distribution (grey line). The dashed grey lines are the asymptotic 95\% standard error bands of a normal distribution.}
	\caption{Densities of estimated parameters when $E||X_t||^4<\infty$. The solid line is the estimated density, and the grey dashed line is the normal distribution.	}
	\label{fig:sim2}
\end{figure}

We simulate $N=10.000$ realizations of \eqref{eq:sim1_1} and \eqref{eq:sim2_2} with $T=10.000$ observations, and estimate $W$ and $A$ using STE. Figure \ref{fig:sim2} contains standardized densities %and QQ-plots 
of $w_1$ and $a_{11}$.  The figure suggests that in this case, the estimator is indeed consistent, but not quite asymptotically normal. Surprisingly, the density of $w_1$ seem to behave almost like a normal distribution, albeit with a heavy left tail, whereas that of $a_{11}$ is clearly non-normal. This is similar to the findings of \citet{Pedersen2014}, who consider variance-targeting in the BEKK-model, and we conclude that $E||X_t||^4<\infty$ is a necessary condition for the joint normality of the ST estimator.

\subsection{Case 3: The DGP does not satisfy the sufficient condition for consistency}
Finally, we consider the case where the DGP only admits a finite mean, $E||X_t||<\infty$. Here we set parameter matrices to
\begin{align}
	V_0 = 
	\begin{pmatrix}
		0.89 & 0.45 \\
       -0.45 & 0.89
	\end{pmatrix}, 	 \ \ \
	W_0 = \begin{pmatrix}
	1.5 \\
	0.46
	\end{pmatrix}, \ \ \ 
	A_0 =
	\begin{pmatrix}
		1.01 & 0 \\
		0 & 0.90
	\end{pmatrix}, \ \ \
	B_0=0_{2\times2}.
	\label{eq:sim3_2}
\end{align}
As $\max(a_{0,ii})<\pi/2$ for $i=1,2$ the DGP is strictly stationary, ergodic and has a finite mean, but does not admit any higher order moments (by Lemma \ref{lemma:moments}). As before, we simulate $N=10.000$ realizations of \eqref{eq:sim1_1} and \eqref{eq:sim3_2} with $T=10.000$ observations, and estimate $W$ and $A$ using STE. Figure \ref{fig:sim3} contains standardized densities % and QQ-plots 
of $w_1$ and $a_{11}$.  Clearly, the estimator is neither consistent nor asymptotically normal when $E||X_t||^2=\infty$, and we conclude that the moment condition in Theorem \ref{theorem:consistency} is necessary for consistency of the estimator.

\begin{figure}[ht]
	\centering
	\includegraphics[width=\textwidth]{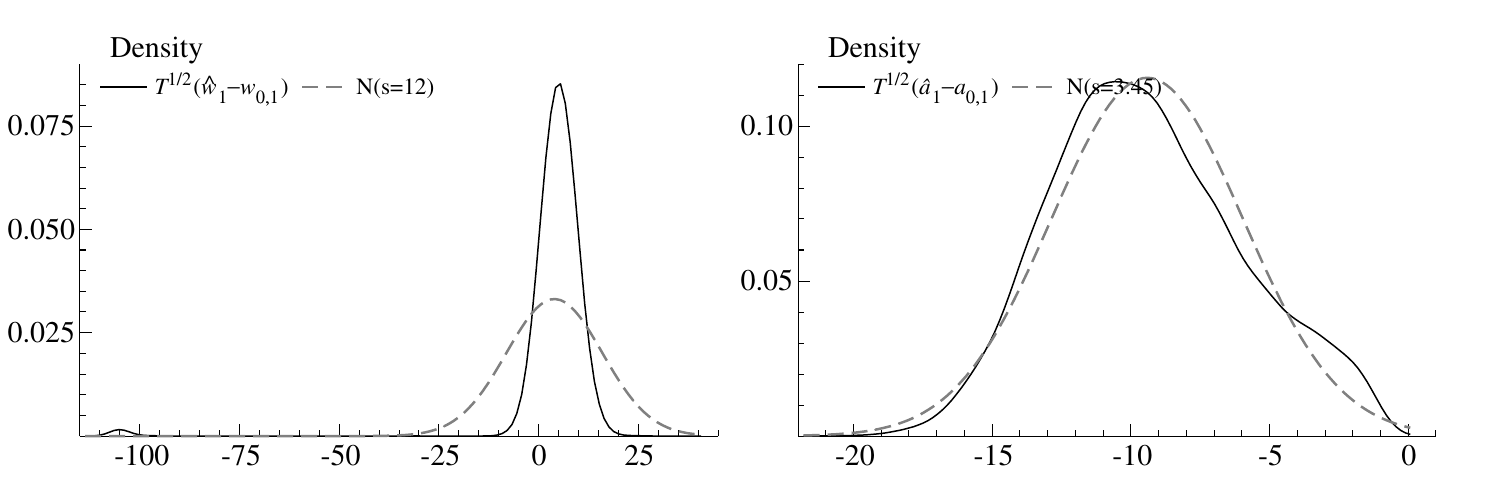}
	%\caption{Densities of estimated parameters when $E||X_t||<\infty$.  In the density plots the solid line is the estimated density, and the grey dashed line is the normal distribution. The QQ-plots compare the quantiles of the estimated parameters with the ones of a normal distribution (grey line). The dashed grey lines are the asymptotic 95\% standard error bands of a normal distribution.}
	\caption{Densities of estimated parameters when $E||X_t||^4<\infty$. The solid line is the estimated density, and the grey dashed line is the normal distribution.	}
	\label{fig:sim3}
\end{figure}

%%%%%%%%%%%%%%%%%%%%%%%%%%%%%%%%%%%%% %%%%%%%%%%%%%%%%%%%%%%%%%%%%%%%%%%%%% %%%%%%%%%%%%%%%%%%%%%%%%%%%%%%%%%%%%% 
%%%%%%%%%%%%%%%%%%%%%%%%%%%%%%%%%%%%%       OLD SIMULATION SECTION END      %%%%%%%%%%%%%%%%%%%%%%%%%%%%%%%%%%%%%
%%%%%%%%%%%%%%%%%%%%%%%%%%%%%%%%%%%%% %%%%%%%%%%%%%%%%%%%%%%%%%%%%%%%%%%%%% %%%%%%%%%%%%%%%%%%%%%%%%%%%%%%%%%%%%% 

\section{Empirical exercise: Portfolio constituents and weights}
\label{appendix:empirical_ex}
% Table generated by Excel2LaTeX from sheet 'Latex table'
\begin{table}[!htb]
  \centering
  \caption{Portfolio constituents and weights}
    \footnotesize{
    \begin{tabular}{l|l|rrrrr}
    \hline
    \hline
    Bloomberg ticker & Company name & \multicolumn{5}{c}{Portfolio weights} \\
     &  & \multicolumn{1}{c}{$P_1$} & \multicolumn{1}{c}{$P_2$} & \multicolumn{1}{c}{$P_3$} & \multicolumn{1}{c}{$P_4$} & \multicolumn{1}{c}{$P_5$} \\
    \hline
    DIS US Equity & Walt Disney Co. &         0.040  &         0.016  &         0.010  &         0.068  & -      0.347  \\
    HD US Equity & Home Depot &         0.040  &         0.002  &         0.085  & -      0.082  & -      0.099  \\
    ABT US Equity & Abbott &         0.040  &         0.042  &         0.071  & -      0.195  & -      1.099  \\
    CVX US Equity & CV Sciences &         0.040  &         0.066  &         0.007  & -      0.038  & -      1.243  \\
    EXC US Equity & Exelon &         0.040  &         0.038  &         0.057  &         0.137  &         0.671  \\
    MCD US Equity & McDonalds &         0.040  &         0.044  &         0.048  & -      0.257  & -      0.122  \\
    MMM US Equity & 3M    &         0.040  &         0.001  &         0.101  & -      0.282  &         0.318  \\
    AAPL US Equity & Apple &         0.040  &         0.045  &         0.086  &         0.159  &         0.621  \\
    UNH US Equity & United Health Group &         0.040  &         0.086  &         0.072  & -      0.118  & -      0.210  \\
    TXN US Equity & Texas Instruments &         0.040  &         0.078  &         0.066  &         0.118  & -      0.342  \\
    JPM US Equity & JPMorgan Chase &         0.040  &         0.067  &         0.113  &         0.138  & -      0.856  \\
    IBM US Equity & IBM   &         0.040  &         0.010  &         0.118  &         0.255  &         0.760  \\
    DVN US Equity & Devon Energy &         0.040  &         0.082  &         0.040  & -      0.059  &         0.861  \\
    GD US Equity & General Dynamics &         0.040  &         0.078  &         0.029  & -      0.083  &         0.606  \\
    CPB US Equity & Campbell Soup Company &         0.040  &         0.015  &         0.095  &         0.180  &         0.970  \\
    PEP US Equity & PepsiCo &         0.040  &         0.058  &         0.008  &         0.156  &         0.383  \\
    MRK US Equity & Merck \& Co. &         0.040  &         0.002  &         0.043  &         0.090  &         0.507  \\
    NKE US Equity & NantKwest &         0.040  &         0.011  &         0.008  &         0.183  & -      0.794  \\
    COST US Equity & Costco &         0.040  &         0.029  &         0.038  &         0.084  &         0.692  \\
    T US Equity & AT\&T &         0.040  &         0.014  &         0.008  &         0.269  & -      0.429  \\
    RF US Equity & Regions Financial Corporation &         0.040  &         0.069  &         0.035  & -      0.098  & -      0.160  \\
    SLB US Equity & Schlumberger &         0.040  &         0.075  &         0.094  &         0.140  & -      0.468  \\
    PG US Equity & Procter \& Gamble &         0.040  &         0.032  &         0.108  &         0.265  &         0.177  \\
    HON US Equity & Honeywell &         0.040  &         0.029  &         0.094  & -      0.097  &         0.410  \\
    HPQ US Equity & Hewlett-Packard &         0.040  &         0.010  &         0.067  &         0.066  &         0.195  \\
    \hline
    \hline
    \end{tabular}%
    }
  \label{tab:pf_w}%
\end{table}%

\end{document}